\documentclass[lettersize,journal]{IEEEtran}
\usepackage{amsmath,amsfonts}
\usepackage{algorithmic}
\usepackage{algorithm}
\usepackage{array}
\usepackage[caption=false,font=normalsize,labelfont=sf,textfont=sf]{subfig}
\usepackage{textcomp}
\usepackage{stfloats}
\usepackage{url}
\usepackage{verbatim}
\usepackage{graphicx}
\def\BibTeX{{\rm B\kern-.05em{\sc i\kern-.025em b}\kern-.08em
    T\kern-.1667em\lower.7ex\hbox{E}\kern-.125emX}}
\usepackage{balance}
\usepackage{amssymb,amsthm}
\usepackage{enumitem}
\usepackage[normalem]{ulem}
\usepackage{booktabs}

\newtheorem{assumption}{Assumption}
\newtheorem{theorem}{Theorem}

\newtheorem{remark}{Remark}
\newtheorem{lemma}{Lemma}

\begin{document}
% \title{How to Use the IEEEtran \LaTeX \ Templates}
% \author{IEEE Publication Technology Department
% \thanks{Manuscript created October, 2020; This work was developed by the IEEE Publication Technology Department. This work is distributed under the \LaTeX \ Project Public License (LPPL) ( http://www.latex-project.org/ ) version 1.3. A copy of the LPPL, version 1.3, is included in the base \LaTeX \ documentation of all distributions of \LaTeX \ released 2003/12/01 or later. The opinions expressed here are entirely that of the author. No warranty is expressed or implied. User assumes all risk.}}

\title{A Control Theory inspired Exploration Method for a Linear Bandit driven by a Linear Gaussian Dynamical System}
\author{Jonathan Gornet, \IEEEmembership{Student Member, IEEE}, Yilin Mo, \IEEEmembership{Senior Member, IEEE},\\ and Bruno Sinopoli, \IEEEmembership{Fellow, IEEE}%
\thanks{J. Gornet and B. Sinopoli are with the Department of Electrical and Systems Engineering, Washington University in St. Louis, St. Louis, MO 63130, USA (email: jonathan.gornet@wustl.edu; bsinopoli@wustl.edu).}
\thanks{Yilin Mo is with the Department of Automation, Tsinghua University,
        Beijing, China 100084 (email: ylmo@tsinghua.edu.cn). }%
}

\maketitle

\begin{abstract}
    The paper introduces a linear bandit environment where the reward is the output of a known Linear Gaussian Dynamical System (LGDS). In this environment, we address the fundamental challenge of balancing exploration---gathering information about the environment---and exploitation---selecting to the action with the highest predicted reward. We propose two algorithms, Kalman filter Upper Confidence Bound (Kalman-UCB) and Information filter Directed Exploration Action-selection (IDEA). Kalman-UCB uses the principle of optimism in the face of uncertainty. IDEA selects actions that maximize the combination of the predicted reward and a term that quantifies how much an action minimizes the error of the Kalman filter state prediction, which depends on the LGDS property called observability. IDEA is motivated by applications such as hyperparameter optimization in machine learning. A major problem encountered in hyperparameter optimization is the large action spaces, which hinder the performance of methods inspired by principle of optimism in the face of uncertainty as they need to explore each action to lower reward prediction uncertainty. To predict if either Kalman-UCB or IDEA will perform better, a metric based on the LGDS properties is provided. This metric is validated with numerical results across a variety of randomly generated environments.     
\end{abstract}

\begin{IEEEkeywords}
Non-stationary Stochastic Multi-armed Bandits, Kalman filters, Stochastic Dynamical Systems
\end{IEEEkeywords}

\section{Introduction}\label{sec:Introduction}

The Stochastic Multi-Armed Bandit (SMAB) problem \cite{lattimore2020bandit} is a well-known framework for modeling decision-making under uncertainty. It has inspired algorithms that address real-world challenges such as hyperparameter optimization in machine learning, which are presented as the Hyperband algorithm introduced in \cite{li2017hyperband} or Bayesian optimization methods as reviewed in \cite{garnett2023bayesian}. In SMAB, there exists a learner and an environment that interact for a set number of iterations called a round. For each round, the learner chooses an action and in response the environment reveals a reward sampled from an unknown distribution dependent on the chosen action. The objective is to maximize the accumulated reward over a horizon length. This framework leads to the problem of \textit{exploration} (how much information the learner gathers about the environment) versus \textit{exploitation} (how much the learner commits to an action that it predicts to return the highest reward).

A well-known strategy for approaching \textit{exploration} versus \textit{exploitation} is the principle of optimism in the face of uncertainty. The principle states that the learner chooses the highest predicted reward within a set confidence level \cite{lattimore2020bandit}. Lai and Robbins \cite{lai1987adaptive} has implemented this principle by introducing the Upper Confidence Bound (UCB) algorithm, which was analyzed by Auer, Cesa-Bianchi, and Fischer in \cite{auer2002finite}. The motivation for the wide-spread use of the principle of optimism in the face of uncertainty such as UCB is its closeness to the regret lower bound (a bound of the lowest obtainable regret for \textit{any} algorithm) \cite{auer2002finite}. The principle was applied by Abbasi-Yadkori, P\'{a}l, and Szepesv\'{a}ri \cite{NIPS2011_e1d5be1c} to linear bandits, which is an environment where the reward is the inner product of a known action vector and an unknown linear parameter.

We introduce a linear bandit where the reward is output of a known Linear Gaussian Dynamical System (LGDS), i.e. the reward is the inner product of an action vector and a system state evolving linearly over time. Our key contribution includes two algorithms, Kalman filter Upper Confidence Bound (Kalman-UCB) and Information filter Directed Exploration Action-selection (IDEA). Both methods use the Kalman filter to predict the reward of the LGDS for each action and are inspired by the UCB algorithm. In Kalman-UCB, the learner selects the action that maximizes a combination of the predicted reward and a term proportional to the prediction error. For IDEA, the learner selects the action that maximizes the combination of the predicted reward and a term that measures how much an action minimizes the error the Kalman filter's state prediction. The motivation for IDEA is based on its applicability to hyperparameter optimization for training reinforcement learning neural networks. Previous results such as Parker-Holder, Nguyen, and Roberts \cite{parker2020provably}, which was based on theoretical developments made by Bogunovic, Scarlett, and Cevher \cite{bogunovic2016time}, have suggested modeling this problem as a LGDS. In this context, the number of actions, or hyperparameter configurations, vastly exceeds the number of rounds. For more details on the derivations and rationale for modeling the hyperparameter optimization problem as a LGDS, see Gornet, Kantaros, and Sinopoli in \cite{gornet2025hypercontrollerhyperparametercontrollerfast}. Finally, we provide a metric for comparing Kalman-UCB and IDEA to predict which method will perform best with respect to the LGDS properties.

The contributions of the paper are as follows. 
\begin{itemize}
    \item We formulate the linear bandit with an unknown parameter vector generated by a LGDS. 
    \item We prove that approaching this SMAB environment as an optimization problem leads to a situation where the optimal prediction and action selection are interconnected, implying that dynamic programming is computationally intractable. 
    \item For evaluating the difficulty of the proposed SMAB environment, we prove a lower bound on performance, which is a measure of the difficulty for consistently selecting the optimal action. 
    \item We propose the methods Kalman filter Upper Confidence Bound (Kalman-UCB) and Information filter Directed Exploration Action-selection (IDEA). Kalman-UCB is an UCB-inspired method. IDEA chooses the action that maximize the sum of the predicted reward and a term proportional to a measure of how much the action will decrease the error of the Kalman filter state prediction. 
    \item We introduce a metric for evaluating each method's relative effectiveness. 
    \item We verify our analysis with numerical results for a set of randomly generated LGDS that have parameters and noise statistics sampled from different distributions: the Gaussian, Cauchy, Uniform, Exponential, and Bernoulli distributions. 
\end{itemize}

The remainder of the paper is structured as follows: Section \ref{sec:Problem Formulation} introduces the linear bandit environment and its associated optimization problem. Optimal estimation and optimal control are reviewed in Subsection \ref{sec:Filtering} and Subsection \ref{sec:Dynamic} respectively. Section \ref{sec:Lower Bound of Environments modeled as LGDS} provides proofs on lower bounds, which are metrics of the linear bandit environment's difficulty. In Section \ref{sec:Optimism-Based Methods}, we introduce optimism-based methods, which are methods that select actions based on the highest predicted reward with a perturbation. Here, we review both Kalman-UCB in Subsection \ref{sec:UCB_Method} and IDEA in Subsection \ref{sec:Observer Method}. Section \ref{sec:Discussion} compares and analyzes both methods. Finally, in Section \ref{sec:Random_Numerical_Comparisons}, we provide numerical results. The paper is concluded in Section \ref{sec:Conclusion}. 

\subsection{Works Related to the Proposed Environment}

For our proposed environment, the reward is the output of a LGDS. When the LGDS is marginally stable or unstable, the reward process for each action becomes non-stationary due to changes in the reward distributions. The state-of-the-art result in non-stationary SMAB was presented by Besbes, Gur, and Zeevi \cite{besbes2014stochastic}, where they constrain the reward distributional changes to a \textit{variational budget}. Our environment is a specific case of the non-stationary bandit, the slowly-varying case, which introduces gradual changes in the reward distributions. In the slowly-varying case, Slivkins \cite{slivkins2008adapting} modeled each action's reward stochastic process as Brownian motion and analyzed well-known bandit algorithms for this environment. This framework has been extended by Chen, Golrezaei, and Bouneffouf \cite{chen2023non} to environments where the rewards follow action-independent $s$-step autoregressive processes. 

The linear bandit problem is well-studied in SMAB, initially proposed by Abe and Long in \cite{abe1999associative}. As mentioned earlier, the current state-of-the-art result is \cite{NIPS2011_e1d5be1c} which uses an UCB-inspired approach. Kuroki et al. have developed a method for addressing cases when either the linear parameter stochastically or adversarially changes \cite{kuroki2024best}, which is relevant to our work given the dynamic nature of the unknown linear parameter. 

Finally, our results are related to the restless bandit problem, which was initially introduced by Whittle in \cite{whittle1988restless}, where each action's reward process is based on an independent discrete-state Markov chain. For every round, the learner observes the reward which is a function of the Markov chain's state. Previous work has used UCB-inspired methods such as \cite{liu2012learning,tekin2012online,ortner2014regret,dai2011non,wang2020restless}, while a Thompson Sampling approach has been introduced in \cite{jung2019regret}. Currently, the state-of-the-art approach is Restless-UCB by Wang, Huang, and Lui \cite{wang2020restless}. Our results share similarities with this bandit environment, as the LGDS is a Markov chain with a continuous state-space while other restless bandit environments have a Markov chain with a discrete state-space. Since the reward for each action is the inner product of the action vector and the LGDS state variable, this structure introduces dependencies between each action's reward process, which are not modeled in current restless bandit models. 

\smallskip
\noindent\textbf{Notation:} For any $x\in\mathbb{R}^n$ and $y\in\mathbb{R}^n$, we have the inner product $\left\langle x, y\right\rangle = x^\top y \in \mathbb{R}$. The distribution $\mathcal{N}\left(\mu,\Sigma\right)$ is a normal distribution with a mean of $\mu \in \mathbb{R}^d$ and a covariance of $\Sigma \in \mathbb{R}^{d \times d}$.

\section{Problem Formulation}\label{sec:Problem Formulation}

In this work, we will be considering a linear bandit where the reward is the output of a known LGDS. For review, the reward $X_t \in \mathbb{R}$ sampled by the environment in the linear bandit has the following expression
\begin{equation}
    X_t = \left\langle a_t, z \right\rangle + \eta_t \nonumber, 
\end{equation}
where $a_t \in \mathcal{A} \subseteq \mathbb{R}^d$ is the learner's chosen action at round $t$, $z \in \mathbb{R}^d$ is the unknown parameter vector, and $\eta_t \in \mathbb{R}$ is zero-mean noise. For this paper, we will assume that the unknown parameter vector $z$ dynamically changes as according to the state variable $z_t$ in a known LGDS, i.e. 
\begin{equation}\label{eq:LGDS}
    \begin{cases}
        z_{t+1} & = \Gamma z_t + \xi_t, ~z_0 \sim \mathcal{N}\left(\mathbf{0},\Sigma_0\right)\\
        X_t & = \left\langle a_t, z_t\right\rangle + \eta_t
    \end{cases}, 
\end{equation}

In the LGDS above, $z_t \in \mathbb{R}^d$ is the system's state and $X_t \in \mathbb{R}$ is the reward. The variable $a_t \in \mathcal{A}$ is the action that the learner chooses. The process noise $\xi_t \in \mathbb{R}^d$ and measurement noise $\eta_t \in \mathbb{R}$ are independent Gaussian distributed, i.e. $\xi_t \sim \mathcal{N}\left(\mathbf{0},Q\right)$ and $\eta_t \sim \mathcal{N}\left(0,\sigma^2\right)$ where $Q \succeq \mathbf{0}$ and $\sigma > 0$. The following assumption is imposed for the action set $\mathcal{A}$:

\begin{assumption}\label{assumption:unit_sphere}
    The set of actions $\mathcal{A}$ is constrained to the unit sphere, i.e. 
    \begin{equation}
        \mathcal{A} \subseteq \mathbb{S}^{d-1} \triangleq \left\{a_t \in \mathbb{R}^d \mid \left\Vert a_t\right\Vert_2 = 1\right\}. 
    \end{equation}
\end{assumption}

Assumption \ref{assumption:unit_sphere} simplifies the considered problem by only analyzing the observability of \eqref{eq:LGDS}. A metric for observability is the \textit{Observability Gramian}, which is defined to be 
\begin{equation}\label{eq:observability_gramian}
    \mathcal{O}\left(\Gamma,t_0,t_1\right) \triangleq \sum_{\tau=t_0}^{t_1} \left(\Gamma^\top\right)^\tau a_\tau a_\tau^\top \Gamma^\tau \in \mathbb{R}^{d \times d}. 
\end{equation}

The system \eqref{eq:LGDS} is observable from round $t_0$ to $t_1$ if the \textit{Observability Gramian} $\mathcal{O}\left(\Gamma,t_0,t_1\right)$ is positive definite. 

\begin{assumption}\label{assumption:controllability}
    The matrix pair $\left(\Gamma,Q^{1/2}\right)$ is controllable. 
\end{assumption}

Assumption \ref{assumption:controllability} is a necessary condition for the existence of the LGDS's \eqref{eq:LGDS} Kalman filter. The intuition behind this assumption is that state vector $z_t$ is constantly perturbed by the process noise $\xi_t$. We will review later the Kalman filter. 

The goal of the learner is to maximize cumulative reward over a horizon of length $n$, i.e. $S_n = \sum_{t=1}^n X_t$. We assume for this work that the horizon length $n$ is known. This leads to the following optimization problem to be solved: 
\begin{equation}\label{eq:optimization_problem}
    \begin{array}{cc}
        \underset{a_1,\dots,a_n \in \mathcal{A}}{\max} & \sum_{t=1}^n \left\langle a_t, z_t \right\rangle\\
        \mbox{ s.t. } & \begin{cases}
        z_{t+1} & = \Gamma z_t + \xi_t, ~z_0 \sim \mathcal{N}\left(0,\Sigma_0\right)\\
        X_t & = \left\langle a_t, z_t\right\rangle + \eta_t
    \end{cases}
    \end{array}. 
\end{equation}
\begin{remark}
    In stochastic multi-armed bandits the metric for performance is regret which is the cumulative expected difference between the highest possible reward $X_t^*$ at each round $t$ and the sampled reward $X_t$ from the learner's chosen action $a_t \in \mathcal{A}$, i.e.
    \begin{equation}\label{eq:regret}
        R_n \triangleq \sum_{t=1} \mathbb{E}\left[X_t^* - X_t\right]. 
    \end{equation}
\end{remark}

\begin{remark}
    We define $a_t^*$ to be the action $a \in \mathcal{A}$ that aligns most closely with the state $z_t$, i.e. 
    \begin{equation}\label{eq:a_star_def}
        a_t^* \triangleq \underset{a \in \mathcal{A}}{\arg\max} ~ \left\langle a, z_t \right\rangle. 
    \end{equation}

    This can be interpreted as the \textit{Oracle} as the learner has full knowledge of the state variable $z_t \in \mathbb{R}^d$. 
\end{remark}

% In the next section, we will explain the difficulty of solving optimization problem \eqref{eq:optimization_problem}. 

Maximizing cumulative reward from the linear bandit with an unknown linear parameter generated by a LGDS is difficult to solve. We will present this difficulty from two different perspectives. In \textit{Perspective 1}, Computational Tractability, we will attempt to solve optimization problem \eqref{eq:optimization_problem} which requires us to use dynamic programming. We will prove that approaching this dynamic programming problem leads to a situation where actions impact both the reward prediction error and the accumulated reward. Therefore, we encounter a nonconvex optimization problem in the dynamic programming problem. In \textit{Perspective 2}, Difficulty of Selecting the Optimal Action, we will analyze the difficulty of the bandit environment by deriving a lower bound for regret \eqref{eq:regret}. We will prove that the optimal method's regret must increase at least linearly, implying that it is difficult even for the optimal method to consistently select the optimal action.

\section{Perspective 1: Computational Intractability}\label{sec:difficulty}

In this section, we will provide insight into the computational intractability of solving the bandit problem optimally. First, we will review optimal estimation/prediction by using the Kalman filter. Next, optimal action selection will be reviewed focusing specifically on dynamic programming. We will then prove how optimal action selection and optimal estimation/prediction are interconnected. This will demonstrate how solving the bandit problem optimally is computationally intractable. In the second perspective, given that computationally intractability of the problem, we will derive a lower bound on regret, 

\subsection{Optimal Estimation: Kalman Filter}\label{sec:Filtering}

Since the state $z_t$ of LGDS \eqref{eq:LGDS} is unknown, then the reward $X_t$ is unknown until action $a_t \in \mathcal{A}$ is selected. Therefore, we propose to predict the state of the system \eqref{eq:LGDS}. Using the state prediction, we can predict which action $a_t \in \mathcal{A}$ will return the highest reward. The optimal 1-step predictor, in the mean squared error sense, of the LGDS's state $z_t$ is the Kalman filter. The Kalman filter (in 1-step predictor form) is written as follows:
% \begin{equation}\label{eq:Kalman_Filter}
%     \begin{cases}
%         \hat{z}_{t+1|t} & = \Gamma \hat{z}_{t|t} \\
%         P_{t+1|t} & = \Gamma P_{t|t} \Gamma^\top + Q \\
%         \hat{z}_{t|t} & = \hat{z}_{t|t-1} + K_t \left(X_t - \left\langle a_t, \hat{z}_{t|t-1}\right\rangle\right) \\
%         K_t & = P_{t|t-1} a_t\left(a_t^\top P_{t|t-1} a_t + \sigma^2\right)^{-1} \\
%         P_{t|t} & = P_{t|t-1} - K_t a_t^\top P_{t|t-1}
%     \end{cases}. 
% \end{equation}
\begin{equation}\label{eq:Kalman_Filter}
    \begin{cases}
        \hat{z}_{t+1|t} & = \Gamma \hat{z}_{t|t} + \Gamma  K_t \left(X_t - \left\langle a_t, \hat{z}_{t|t-1}\right\rangle\right) \\
        P_{t+1|t} & = g\left(P_{t|t-1},a_t\right) \\
        K_t & = P_{t|t-1} a_t\left(a_t^\top P_{t|t-1} a_t + \sigma^2\right)^{-1}
    \end{cases}. 
\end{equation}
where $g\left(P_{t|t-1},a\right)$ is defined to be 
\begin{multline}\label{eq:g_definition}
    g\left(P_{t|t-1},a_t\right) \triangleq \Gamma P_{t|t-1} \Gamma^\top + Q \\ - \Gamma P_{t|t-1} a_t \left(a_t^\top P_{t|t-1} a_t + \sigma\right)^{-1} a_t^\top P_{t|t-1} \Gamma^\top .
\end{multline}
    
The estimate of the state $z_t$ is defined to be $\hat{z}_{t|t} \triangleq \mathbb{E}\left[z_t \mid \mathcal{F}_t\right]$, where $\mathcal{F}_t$ is the sigma algebra generated by previous observations $X_0,\dots,X_t$. The matrix $P_{t|t-1}$ is the error covariance matrix of the state estimate $\hat{z}_{t|t-1}$, i.e. the covariance of $e_{t|t-1} \triangleq z_t - \hat{z}_{t|t-1}$. The error covariance matrix $P_{t|t-1}$ converges if the matrix pair $\left(\Gamma,a_t^\top\right)$ is detectable and $\left(\Gamma,Q^{1/2}\right)$ is controllable, where the controllability assumption is imposed in Assumption \ref{assumption:controllability}. The following lemma provides known facts about the Kalman filter \cite{sinopoli2005optimal}:

\begin{lemma}\label{lemma:Kalman_Facts}
    The following facts are true for the Kalman filter \eqref{eq:Kalman_Filter}:
    \begin{itemize}
        \item $\mathbb{E}\left[e_{t|t-1}^\top \hat{z}_{t|t-1}\mid \mathcal{F}_{t-1}\right] = 0$. 
        \item $\mathbb{E}\left[z_t^\top S z_t\mid \mathcal{F}_{t-1}\right] = \hat{z}_{t|t-1}^\top S \hat{z}_{t|t-1} + \mbox{tr}\left(S P_{t|t-1}\right)$ for all $S \succeq 0$.
        \item $\mathbb{E}\left[\mathbb{E}\left[z_t \mid \mathcal{F}_t\right]\mid \mathcal{F}_{t-1}\right] = \mathbb{E}\left[z_t \mid \mathcal{F}_{t-1}\right]$. 
    \end{itemize}
\end{lemma}

\subsection{Optimal Control: Dynamic Programming}\label{sec:Dynamic}

A common approach in optimal control theory for solving optimization problems \eqref{eq:optimization_problem} is to use a dynamic programming approach. The value function $V_t\left(z_t\right)$ is defined as follows 
\begin{equation}\label{eq:dynamic_programming}
    \begin{cases}
        V_n \left(z_n\right) & \triangleq \underset{a \in \mathcal{A}}{\max} ~ \mathbb{E}\left[\left\langle a, z_n \right\rangle \mid \mathcal{F}_{n-1}\right] \\
        V_t \left(z_t\right)& = \underset{a_t \in \mathcal{A}}{\max} ~ \mathbb{E}\left[\left\langle a_t, z_t \right\rangle + V_{t+1} \left(z_{t+1}\right)\mid \mathcal{F}_{t-1}\right]
    \end{cases},
\end{equation}
where $t=n-1,n-2,\dots,1$. Dynamic programming theory states that $V_1\left(z_1\right)$ is the optimal value of the optimization problem \eqref{eq:optimization_problem} \cite{bertsekas2012dynamic}. 

When using dynamic programming \eqref{eq:dynamic_programming} for solving \eqref{eq:optimization_problem}, it is proven in the theorem below that the Separation Principle does not hold. The Separation Principle in stochastic optimal control states that optimal estimation (the Kalman filter) and optimal control (solving optimization problem \eqref{eq:optimization_problem}) can be treated as separate problems \cite{georgiou2013separation}. However, the following theorem proves that optimal control and estimation are interconnected.

\begin{theorem}\label{theorem:Separation_Principle}
    Let there be the value function and its iteration defined in \eqref{eq:dynamic_programming}. The $n-1$ step of the value function iteration is a nonlinear function of the error covariance matrix $P_{n-1|n-2}$ and the expectation $\mathbb{E}\left[ \left\Vert z_n \right\Vert_2^2 \mid \mathcal{F}_{n-1}\right]$, which has the following expression:
    % \begin{multline}\label{eq:n-1_step}
    %     V_{n-1} \left(z_{n-1}\right) = \underset{a \in \mathcal{A}}{\max} ~ \left\langle a, \hat{z}_{n-1|n-2} \right\rangle \\ +\mathbb{E}\Bigg[\Bigg( - \mbox{tr}\Bigg(\Gamma P_{n-1|n-2}\Gamma^\top + Q - \frac{\Gamma P_{n-1|n-2} a a^\top P_{n-1|n-2} \Gamma^\top}{a^\top P_{n-1|n-2} a + \sigma^2} \\
    %     + \mathbb{E}\left[ \left\Vert z_n \right\Vert_2^2 \mid \mathcal{F}_{n-1}\right]\Bigg)\Bigg)^{1/2} \mid \mathcal{F}_{n-2}\Bigg]. 
    % \end{multline}
    \begin{multline}\label{eq:n-1_step}
        V_{n-1} \left(z_{n-1}\right) = \underset{a \in \mathcal{A}}{\max} ~ \left\langle a, \hat{z}_{n-1|n-2} \right\rangle + \\ \mathbb{E}\left[\sqrt{\mathbb{E}\left[ \left\Vert z_n \right\Vert_2^2 \mid \mathcal{F}_{n-1}\right] -\mbox{tr}\left(g\left(P_{n-1|n-2},a\right)
        \right)} \mid \mathcal{F}_{n-2}\right]. 
    \end{multline}
\end{theorem}

\begin{proof}
    The solution of the first iteration in the dynamic programming approach \eqref{eq:dynamic_programming} is the action $a \in \mathcal{A}$ that aligns most closely with the state prediction $\hat{z}_{n|n-1}$: 
    \begin{align}
        V_n\left(z_n\right) & = \underset{a \in \mathcal{A}}{\max} ~ \mathbb{E}\left[\left\langle a, z_n \right\rangle \mid \mathcal{F}_{n-1}\right] \nonumber\\
        & = \underset{a \in \mathcal{A}}{\max} ~ \left\langle a, \hat{z}_{n|n-1} \right\rangle. \label{eq:greedy_inspiration} 
    \end{align}

    The action $a\in \mathcal{A}$ that maximizes the function $V_n\left(z_n\right)$ is therefore 
    \begin{equation}\label{eq:myopic_method}
        \underset{a \in \mathcal{A}}{\arg\max} ~ \left\langle a, \hat{z}_{n|n-1} \right\rangle = \frac{\hat{z}_{n|n-1}}{\left\Vert \hat{z}_{n|n-1}\right\Vert_2}, 
    \end{equation}
    providing the expression of the function $V_n\left(z_n\right)$:
    \begin{equation}
        V_n\left(z_n\right) = \left\Vert \hat{z}_{n|n-1}\right\Vert_2. \nonumber 
    \end{equation}

    The second iteration of the dynamic programming approach \eqref{eq:dynamic_programming} using \eqref{eq:myopic_method} has the following expression 
    \begin{align}
        V_{n-1} \left(z_{n-1}\right) & = \underset{a \in \mathcal{A}}{\max} ~ \mathbb{E}\left[\left\langle a, z_{n-1} \right\rangle + V_n \left(z_n\right)\mid \mathcal{F}_{n-2}\right] \nonumber \\
        & = \underset{a \in \mathcal{A}}{\max} ~ \mathbb{E}\left[\left\langle a, z_{n-1} \right\rangle \mid \mathcal{F}_{n-2}\right] \nonumber \\
        & ~~ +\mathbb{E}\left[\left\Vert \hat{z}_{n|n-1}\right\Vert_2 \mid \mathcal{F}_{n-2}\right], \nonumber   
    \end{align}
    \begin{multline}
        \Rightarrow V_{n-1} \left(z_{n-1}\right) \overset{(a)}{=} \underset{a \in \mathcal{A}}{\max} ~ \mathbb{E}\left[\left\langle a, z_{n-1} \right\rangle \mid \mathcal{F}_{n-2}\right] \nonumber \\
        \\ +\mathbb{E}\left[\sqrt{\mathbb{E}\left[ \left\Vert z_n \right\Vert_2^2 \mid \mathcal{F}_{n-1}\right] - \mbox{tr}\left(P_{n|n-1}\right)} \mid \mathcal{F}_{n-2}\right]. \nonumber
    \end{multline}
    where at the $n-1$ step conditioned on $\mathcal{F}_{n-2}$ we arrive at expression \eqref{eq:n-1_step}. For $(a)$, we used the fact that $\left\Vert \hat{z}_{n|n-1}\right\Vert_2^2 = \mathbb{E}\left[ \left\Vert z_n \right\Vert_2^2 \mid \mathcal{F}_{n-1}\right] -  \mbox{tr}\left(P_{n|n-1}\right)$ which is proven in Lemma \ref{lemma:Kalman_Facts}.
\end{proof}

Theorem \ref{theorem:Separation_Principle} proves two important details about using dynamic programming for solving optimization problem \eqref{eq:optimization_problem}. First, at iteration $n-1$, the value function consists of an optimization problem where the error covariance matrix $P_{n-1|n-2}$ is a function of the action $a \in \mathcal{A}$. This implies that the chosen action directly affects estimation, failing to separate the problems of optimal control and optimal estimation. Second, the iteration \eqref{eq:n-1_step} is a nonlinear, nonconvex function of the action $a \in \mathcal{A}$ where in the general case does not have a closed-form analytic solution. Therefore, continuing the iterations of $V_t\left(z_t\right)$, $t = n-1,\dots,1$ does not provide a closed-form analytic expression. Since computing the optimal control is computationally complex, we will first analyze the regret lower bound, which provides a bound on what is the best a learner can accomplish.

\section{Perspective 2: Difficulty of Selecting the Optimal Action}\label{sec:Lower Bound of Environments modeled as LGDS}

For this section, we provide the lower bound of regret \eqref{eq:regret} for SMAB environments modeled as LGDS. This provides a measure of the environment's difficulty by tracking how hard it is to consistently select the optimal action. The approach we use is to use the \textit{principle of optimality} \cite{bertsekas2012dynamic}, i.e. the optimal policy that solves the optimization problem defined as \eqref{eq:optimization_problem} for $n$ steps is also the optimal policy for any length $n' < n$. Upper bounding the optimal value for the dynamic programming problem provides a lower bound for regret $R_n$. There are two lower bounds that are provided in this section. The first lower bounds is for actions on the unit sphere, i.e. $a \in \mathcal{A} \triangleq \left\{a \in \mathbb{R}^d \mid \left\Vert a\right\Vert_2 = 1\right\}$. This bounds gives intuition to what a policy close to the optimal policy may look like. The next lower bounds is for a discrete number of actions $a \in \mathcal{A}$, $\left\vert \mathcal{A}\right\vert = k$. First, the theorem below provides the lower bound of regret for the actions on the unit sphere. 

\begin{theorem}\label{theorem:lower_bound}
    Let there be the continuous action set $\mathcal{A} = \left\{a \mid \left\Vert a \right\Vert_2 = 1, a \in \mathbb{R}^d\right\}$. Assume that there exists a $P'$ such that $P_{t|t-1} \succeq P'$ for any $t = 1,2,\dots,n$. The lower bound for regret for the action set $\mathcal{A}$ is 
    \begin{equation}\label{eq:regret_lower_bound}
        R_n \geq \sum_{t=1}^n \mathbb{E}\left[\sqrt{\nu_t^\top Z_t \nu_t}\right] - \mathbb{E}\left[\sqrt{\hat{\nu}_t^\top \left(Z_t - P'\right)\hat{\nu}_t}\right]. 
    \end{equation}
    where $Z_t$ is defined to be 
    \begin{equation}
        Z_t \triangleq \mathbb{E}\left[z_t z_t^\top \right] \label{eq:z_covariance}, 
    \end{equation}
    and $\nu_t, \hat{\nu}_t \sim \mathcal{N}\left(\mathbf{0}, I_d\right)$. If $\rho\left(\Gamma\right) < 1$ and $Z_t \rightarrow Z$ and $P_{t|t-1}\rightarrow P$, then regret is satisfies the following inequality
    \begin{equation}\label{eq:regret_lower_bound_inequality_final}
        R_n \geq \sum_{t=1}^n \mathbb{E}\left[\sqrt{\nu_t^\top Z \nu_t}\right] - \mathbb{E}\left[\sqrt{\hat{\nu}_t^\top \left(Z - P\right)\hat{\nu}_t}\right] . 
    \end{equation}
\end{theorem}
\begin{proof}
    Let regret $R_n \triangleq \mathbb{E}\left[\sum_{t=1}^n X_t^* - X_t\right]$ where $X_t^* \triangleq \max_{a \in \mathcal{A}}\left\langle a, z_t\right\rangle$. Recall that we can express the regret as the following
    \begin{align}
        R_n & = \sum_{t=1}^n \mathbb{E}\left[X_t^* - X_t \right] \nonumber \\
        & = \sum_{t=1}^n \max_{a \in \mathcal{A}} \left\langle a, z_t \right\rangle - \left\langle a, z_t \right\rangle \nonumber. 
    \end{align}

    To lower bound the regret, we know that the optimal policy $\pi$ that minimizes regret follows the \textit{principle of optimality} \cite{bertsekas2012dynamic}. If we find the optimal value $\mathbb{E}_{\pi_t}\left[X_t \right]$ for each round $t$, then the summation of optimal values $\mathbb{E}_{\pi_t}\left[X_t \right]$ from $t = 1,2,\dots,n$ gives $\sum_{t=1}^n \mathbb{E}_{\pi_t}\left[X_t \right]$ which is optimal. Therefore, by upper bounding $\mathbb{E}_{\pi_t}\left[X_t \right]$, we lower bound the regret. Consider the dynamic programming problem where $\hat{V}_n\left(\hat{z}_{n|n-1}\right) = \max_{a \in \mathcal{A}}\mathbb{E}\left[\left\langle a, \hat{z}_{n|n-1} \right\rangle \right]$ that has the following iteration 
    \begin{equation}
        \hat{V}_t\left(\hat{z}_{t|t-1}\right) = \max_{a \in \mathcal{A}} ~\mathbb{E}\left[\hat{V}_{t+1}\left(\hat{z}_{t+1|t}\right) + \left\langle a, \hat{z}_{t|t-1}\right\rangle \right] \nonumber. 
    \end{equation}
    
    Based on Theorem \ref{theorem:Separation_Principle}, we can observe that 
    \begin{align}
        \hat{V}_n\left(\hat{z}_{n|n-1}\right) & = \max_{a \in \mathcal{A}}\mathbb{E}\left[\left\langle a, \hat{z}_{n|n-1} \right\rangle \right] \nonumber \\
        & \overset{(a)}{=} \mathbb{E}\left[\left\Vert \hat{z}_{n|n-1} \right\Vert_2\right],  \label{eq:value_function_lower}
    \end{align}    
    where for $(a)$ we used \eqref{eq:myopic_method}. Continuing the iteration for $t= n-1$ provides
    \begin{equation} 
        \hat{V}_{n-1}\left(\hat{z}_{n-1|n-2}\right) = \max_{a \in \mathcal{A}} \mathbb{E}\left[\hat{V}_n\left(\hat{z}_{n|n-1}\right) + \left\langle a, \hat{z}_{n-1|n-2}\right\rangle\right] \nonumber 
    \end{equation}
    \begin{multline}
        \Rightarrow \hat{V}_{n-1}\left(\hat{z}_{n-1|n-2}\right) = \max_{a \in \mathcal{A}} \mathbb{E}\left[\left\Vert \hat{z}_{n|n-1} \right\Vert_2\right] \\ +  \mathbb{E}\left[\left\langle a, \hat{z}_{n-1|n-2}\right\rangle\right] \label{eq:cont_lower_1}, 
    \end{multline}

    Based on Theorem \ref{theorem:Separation_Principle}, the term $\mathbb{E}\left[\left\Vert \hat{z}_{n|n-1} \right\Vert_2\right]$ is dependent on $a \in \mathcal{A}$. Therefore, we will use an upper bound  of $\mathbb{E}\left[\left\Vert \hat{z}_{n|n-1} \right\Vert_2\right]$ that is independent on $a \in \mathcal{A}$. First, since $\hat{z}_{t|t-1} = \hat{Z}_{t|t-1}^{1/2} \hat{\nu}$ where $\hat{Z}_{t|t-1} \triangleq \mathbb{E}\left[\hat{z}_{t|t-1}\hat{z}_{t|t-1}^\top\right]$ and $\hat{\nu}_t \sim \mathcal{N}\left(\mathbf{0},I_d\right)$, then $\mathbb{E}\left[\left\Vert \hat{z}_{n|n-1} \right\Vert_2\right]$ can be expressed as 
    \begin{equation}
        \mathbb{E}\left[\left\Vert \hat{z}_{n|n-1} \right\Vert_2\right] = \mathbb{E}\left[\sqrt{\hat{\nu}_t^\top \hat{Z}_{n|n-1} \hat{\nu}_t}\right] \label{eq:cont_lower_2}.
    \end{equation}
    
    Since $z_t = \hat{z}_{t|t-1} + e_{t|t-1}$, then we can express $\hat{Z}_{t|t-1}$ using the following:
    \begin{align}
        Z_t & = \mathbb{E}\left[z_t z_t^\top \right] \nonumber\\
        & = \mathbb{E}\left[(\hat{z}_{t|t-1} + e_{t|t-1}) (\hat{z}_{t|t-1} + e_{t|t-1})^\top\right] \nonumber\\
        & = \mathbb{E}\left[\hat{z}_{t|t-1}\hat{z}_{t|t-1}^\top \right] + \mathbb{E}\left[\hat{z}_{t|t-1} e_{t|t-1}^\top \right]  + \mathbb{E}\left[e_{t|t-1}\hat{z}_{t|t-1}^\top \right] \nonumber \\
        & ~~+ \mathbb{E}\left[e_{t|t-1} e_{t|t-1}^\top \right] \nonumber\\
        & \overset{(b)}{=}  \hat{Z}_{t|t-1} + P_{t|t-1} \nonumber,
    \end{align}
    \begin{equation}
        \Rightarrow \hat{Z}_{t|t-1} = Z_t - P_{t|t-1} \label{eq:Kalman_Filter_orthogonality_principle}, 
    \end{equation}
    where in $(b)$ we used Lemma \ref{lemma:Kalman_Facts} and 
    \begin{equation}
        \mathbb{E}\left[\hat{z}_{t|t-1} e_{t|t-1}^\top \right] = \mathbb{E}\left[\mathbb{E}\left[\hat{z}_{t|t-1} e_{t|t-1}^\top \mid \mathcal{F}_{t-1} \right]\right] = \mathbf{0}, 
    \end{equation}
    
    Therefore, using \eqref{eq:cont_lower_2}, \eqref{eq:Kalman_Filter_orthogonality_principle}, and the detail that $P_{t|t-1} \succeq P'$ for any $t = 1,2,\dots,n$, $\mathbb{E}\left[\left\Vert \hat{z}_{n|n-1} \right\Vert_2\right]$ has the following upper bound 
    \begin{equation} 
        \mathbb{E}\left[\left\Vert \hat{z}_{n|n-1} \right\Vert_2\right] = \mathbb{E}\left[\sqrt{\hat{\nu}_t^\top \left(Z_n - P_{n|n-1}\right)\hat{\nu}_t}\right] \nonumber 
    \end{equation}
    \begin{equation} 
        \Rightarrow \mathbb{E}\left[\left\Vert \hat{z}_{n|n-1} \right\Vert_2\right] \leq \mathbb{E}\left[\sqrt{\hat{\nu}_t^\top \left(Z_n - P'\right)\hat{\nu}_t}\right] \nonumber 
    \end{equation}

    The above implies that \eqref{eq:cont_lower_1} has the following upper bound where now the upper bound of $\mathbb{E}\left[\left\Vert \hat{z}_{n|n-1} \right\Vert_2\right]$ is independent of $a \in \mathcal{A}$:
    \begin{multline}
        \hat{V}_{n-1}\left(\hat{z}_{n-1|n-2}\right) \leq \max_{a \in \mathcal{A}} \mathbb{E}\left[\left\langle a, \hat{z}_{n-1|n-2}\right\rangle\right] \\ +\mathbb{E}\left[\sqrt{\hat{\nu}_t^\top \left(Z_n - P'\right)\hat{\nu}_t}\right]  \label{eq:cont_lower_3}, 
    \end{multline}

    Continuing the iteration for $t = n-1,n-2,\dots,0$ provides 
    \begin{multline}
        \hat{V}_{0}\left(\hat{z}_{0|-1}\right) \leq \max_{a \in \mathcal{A}} \mathbb{E}\left[\left\langle a, \hat{z}_{0|-1}\right\rangle\right] \\ + \sum_{t=1}^n \mathbb{E}\left[\sqrt{\hat{\nu}_t^\top \left(Z_t - P'\right)\hat{\nu}_t}\right]. 
    \end{multline}

    The above leads to the following lower bound of regret:
    \begin{align}
        R_n &= \mathbb{E}\left[\sum_{t=1}^n X_t^* - X_t \right]\nonumber \\
        & = - \hat{V}_0\left(\hat{z}_{0|-1}\right) + \mathbb{E}\left[\sum_{t=1}^n X_t^* \right] \nonumber \\
        & \geq - \sum_{t=0}^n \mathbb{E}\left[\sqrt{\hat{\nu}_t^\top \left(Z_t - P'\right)\hat{\nu}_t}\right] + \sum_{t=1}^n \mathbb{E}\left[\left\Vert z_t \right\Vert_2\right] \nonumber, 
    \end{align}
    leading to inequality \eqref{eq:regret_lower_bound_inequality_final}. 
    
\end{proof}

Based on the Theorem \ref{theorem:lower_bound}, the best a learner can do is dependent on lowest obtainable error covariance matrix $P'$. Therefore, the lower bound states implicitly that the error is accumulating linearly. 

The following theorem provides the lower bound for a finite number of actions $\mathcal{A}$. This offers deeper insight into how the linear accumulation of the error is the cause of a linear increasing lower bound. First, we provide the \textit{Kalman Oracle Action-selection}, Algorithm \ref{alg:Kalman_Oracle}, which utilizes the following \textit{Kalman Oracle}
\begin{equation}
    \begin{cases}
        \tilde{z}_{t+1} & = \Gamma \tilde{z}_t + \Gamma K \left(\mathbf{X}_t - C_{\mathcal{A}} \tilde{z}_t\right)\\
        \tilde{\mathbf{X}}_t & = C_{\mathcal{A}} \tilde{z}_t 
    \end{cases}\label{eq:LGDS_Oracle_Kalman_Filter}.
\end{equation}

The state prediction $\tilde{z}_t \triangleq \mathbb{E}\left[z_t \mid \mathcal{G}_{t-1}\right]$ and $\mathcal{G}_{t-1}$ is the sigma algebra of $\mathbf{X}_0,\dots,\mathbf{X}_{t-1}$. The observation $\mathbf{X}_t \in \mathbb{R}^k$ a vector of the rewards for each action $a \in \mathcal{A}$, i.e. the output of the following LGDS: 
\begin{equation}
    \begin{cases}
        z_{t+1} & = \Gamma z_t + \xi_t \\
        \mathbf{X}_t & = C_{\mathcal{A}} z_t + \begin{pmatrix}
            \eta_t^{(1)} \\
            \vdots \\
            \eta_t^{(k)}
        \end{pmatrix}
    \end{cases} \label{eq:LGDS_Oracle}.
\end{equation}

Finally, $C_{\mathcal{A}}$ in \eqref{eq:LGDS_Oracle_Kalman_Filter} and \eqref{eq:LGDS_Oracle} and $K \in \mathbb{R}^{d \times k}$ are defined to be 
\begin{align}
    C_{\mathcal{A}} & \triangleq \begin{pmatrix}
        a_1 & \dots & a_k
    \end{pmatrix}^\top \in \mathbb{R}^{k \times d} \label{eq:c_mathcal_a} \\
    K & \triangleq P C_{\mathcal{A}}^\top \left(C_{\mathcal{A}}P C_{\mathcal{A}}^\top  + \sigma^2 I_k \right)^{-1} \nonumber \\
    P_{\mathcal{A}} & = \Gamma P_{\mathcal{A}} \Gamma^\top + Q \nonumber \\
    & ~~~~~~~~~~~~~- \Gamma P_{\mathcal{A}} C_{\mathcal{A}}^\top \left(C_{\mathcal{A}}P_{\mathcal{A}} C_{\mathcal{A}}^\top  + \sigma^2 I_k \right)^{-1}C_{\mathcal{A}}P_{\mathcal{A}}\Gamma^\top,  \nonumber
\end{align}
where $P_{\mathcal{A}}$ is the steady-state error covariance matrix of the Kalman filter state prediction $\tilde{z}_{t|t-1}$ in \eqref{eq:LGDS_Oracle_Kalman_Filter}. In \textit{Kalman Oracle Action-selection}, there exists an \verb|Action Selection|, \verb|Observation|, and \verb|Update|. In \verb|Action Selection|, \textit{Kalman Oracle Action-selection} selects actions $a \in \mathcal{A}$ such that 
\begin{equation}
    \tilde{a}_t \triangleq \underset{a \in \mathcal{A}}{\arg\max} \left\langle a, \tilde{z}_t \right\rangle \label{eq:a_tilde_def}, 
\end{equation}
or the action $a \in \mathcal{A}$ that aligns most closely with the Kalman filter posed in \eqref{eq:LGDS_Oracle_Kalman_Filter} state prediction $\tilde{z}_{t|t-1}$. The \textit{Kalman Oracle Action-selection} then observes $\mathbf{X}_t$ in the \verb|Observation| step from \eqref{eq:LGDS_Oracle} and updates $\tilde{z}_{t|t-1}$ in \eqref{eq:LGDS_Oracle_Kalman_Filter} for the \verb|Update| step. Based on the formulation of the \textit{Kalman Oracle}, it is not applicable to our setting since the learner can only observe the reward $X_t$ for the selected action $a_t \in \mathcal{A}$. However, we use this algorithm as a baseline for analyzing the difficulty of selecting the optimal action $a_t^* \in \mathcal{A}$ \eqref{eq:a_star_def}. 
\begin{algorithm}[!t]
\caption{\textit{Kalman Oracle Action-selection}}\label{alg:Kalman_Oracle}
 \begin{algorithmic}[1]
\STATE \textbf{Input}: $\Gamma$, $\mathcal{A}$, $Q$, $\sigma$, $\Sigma_0$, $z_0$
\FOR{$t=1,2,\dots,n$}
    \STATE \verb|/* Action Selection */|
    \STATE $a_t = \underset{a \in \mathcal{A}}{\arg\max} \left\langle a, \tilde{z}_{t|t-1} \right\rangle $  
    \STATE \verb|/* Observation */|
    \STATE Observe $\mathbf{X}_t = C_{\mathcal{A}} z_t + \begin{pmatrix}
            \eta_t^{(1)} \\
            \vdots \\
            \eta_t^{(k)}
        \end{pmatrix}$ 
    \STATE \verb|/* Update */|
    \STATE Update $\tilde{z}_{t+1}$ in the \textit{Kalman Oracle} \eqref{eq:LGDS_Oracle_Kalman_Filter}
\ENDFOR
\end{algorithmic}
\end{algorithm}

\begin{lemma}\label{lemma:lower_bound_optimal_discrete}
    Let there be the following LGDS \eqref{eq:LGDS_Oracle} and its associated \textit{Kalman Oracle} \eqref{eq:LGDS_Oracle_Kalman_Filter}. The optimal policy for maximizing the sum $\sum_{t=1}^n X_t$ using the state prediction $\tilde{z}_t$ is \eqref{eq:a_tilde_def} which satisfies the Separation Principle. 
\end{lemma}
\begin{proof}
    We know that the optimal policy $\pi$ that minimizes regret follows the \textit{principle of optimality} \cite{bertsekas2012dynamic}. If we find the optimal value $\mathbb{E}_{\pi_t}\left[X_t \right]$ for each round $t$, then the summation of optimal values $\mathbb{E}_{\pi_t}\left[X_t \right]$ from $t = 1,2,\dots,n$ gives $\sum_{t=1}^n \mathbb{E}_{\pi_t}\left[X_t \right]$ which is optimal. Consider the dynamic programming problem where $\tilde{V}_n\left(z_n\right) = \max_{a \in \mathcal{A}}\mathbb{E}\left[\left\langle a, z_n \right\rangle\mid \mathcal{G}_{n-1}\right]$ that has the following iteration 
    \begin{equation}
        \tilde{V}_t\left(z_t\right) = \max_{a \in \mathcal{A}} ~\tilde{V}_{t+1}\left(z_{t+1} \right) + \mathbb{E}\left[\left\langle a, z_t\right\rangle \mid \mathcal{G}_t\right]. 
    \end{equation}
    
    We can observe that 
    \begin{align}
        \tilde{V}_n\left(z_n\right) & = \max_{a \in \mathcal{A}}\mathbb{E}\left[\left\langle a, z_n \right\rangle\mid \mathcal{G}_{n-1}\right] \nonumber \\
        & = \left\langle \tilde{a}_n, \Tilde{z}_{n} \right\rangle \nonumber. 
    \end{align}    
    \begin{align} 
        \tilde{V}_{n-1}\left(z_{n-1}\right) & = \max_{a \in \mathcal{A}} ~\mathbb{E}\left[\tilde{V}_{n}\left(z_{n} \right) + \left\langle a, z_{n-1}\right\rangle \mid \mathcal{G}_{n-2}\right] \nonumber \\
        & = \max_{a \in \mathcal{A}} ~\left\langle \tilde{a}_n, \Tilde{z}_{n} \right\rangle  + \mathbb{E}\left[\left\langle a, z_{n-1}\right\rangle \mid \mathcal{G}_{n-2}\right] \nonumber, 
    \end{align}
    \begin{equation}\label{eq:n-1_value_bound_discrete}
        \Rightarrow \tilde{V}_{n-1}\left(z_{n-1}\right) = \left\langle \tilde{a}_n, \Tilde{z}_{n} \right\rangle + \left\langle \tilde{a}_{n-1}, \Tilde{z}_{n-1} \right\rangle, 
    \end{equation}

    Based on above, we satisfy the Separation Principle. Therefore, we can continue the iteration to get the optimal value $\tilde{V}_0\left(z_0\right)$ which is 
    \begin{equation}
        \tilde{V}_0\left(z_0\right) = \sum_{t=1}^n \left\langle \tilde{a}_t, \tilde{z}_t\right\rangle. 
    \end{equation}

    Therefore, the optimal policy for maximizing $\sum_{t=1}^n X_t$ using $\tilde{z}_t$ is \eqref{eq:a_tilde_def}.     
\end{proof}

Lemma \ref{lemma:lower_bound_optimal_discrete} states that if we can observe all the rewards for each action, then the Separation Principle applies. Therefore, we can compute the optimal policy for each given round $t$, which leads to an one-step action selection method. Using the policy provided in Lemma \ref{lemma:lower_bound_optimal_discrete}, we can prove the lower bound for the discrete action set. 
\begin{theorem}\label{theorem:lower_bound_discrete}
    Let there regret $R_n$ \eqref{eq:regret}. The lower bound for regret $R_n$ is the following inequality 
    \begin{equation}\label{eq:regret_lower_bound_discrete}
        R_n \geq n \sum_{i \in [k]} \sum_{j \in [k]} \sqrt{\frac{2\left(a_j - a_i \right)^\top Z \left(a_j - a_i\right)}{\mbox{tr}\left(\Psi_{i|j}\right)^{2k-2} \left\vert \Tilde{\Sigma}_{i|j}\right\vert }},
    \end{equation}
    where $\Tilde{\Sigma}_{i|j},\Psi_{i|j}$ are defined to be 
    \begin{align}
        \Tilde{\Sigma}_{i|j} & \triangleq A_i \Tilde{Z} A_i^\top - A_i\Tilde{Z} A_j^\top \left(A_j Z A_j^\top \right)^{-1} A_j \Tilde{Z} A_i^\top \label{eq:tilde_sigma_ij}\\
        \Psi_{i|j} & \triangleq \begin{pmatrix}             \Sigma_{i|j}^{-1} & \Sigma_{i|j}^{-1}\Pi_{i|j} \\             \Pi_{i|j}^\top \Sigma_{i|j}^{-1} & \Pi_{i|j}^\top \Sigma_{i|j}^{-1}\Pi_{i|j}         
        \end{pmatrix}\label{eq:psi_ij},
    \end{align}
    which are based on the following defined terms
    \begin{align}
        \begin{pmatrix}
            A_i \left(z_t - e_{t|t-1}\right) \\
            A_j z_t
        \end{pmatrix} & \sim \mathcal{N}\left(\mathbf{0}, \Sigma_{i,j}\right) \nonumber \\
        \Sigma_{i,j} & \triangleq \begin{pmatrix}
            A_i \Tilde{Z} A_i^\top & A_i\Tilde{Z} A_j^\top \\
            A_j \Tilde{Z} A_i^\top & A_j Z A_j^\top 
        \end{pmatrix} \label{eq:oracle_covariance} \\
        A_i \left(z_t - e_{t|t-1}\right) \mid A_j z_t & \sim \mathcal{N}\left( \Pi_{i|j}z_t , \Tilde{\Sigma}_{i|j}\right) \nonumber \\
        A_i & \triangleq \begin{pmatrix}
            a_i - a_1' & \dots & a_i - a_{k-1}'
        \end{pmatrix}^\top \nonumber  \\
        \Pi_{i|j} & \triangleq A_i\Tilde{Z} A_j^\top \left(A_j Z A_j^\top \right)^{-1} A_j  \nonumber . 
    \end{align}
\end{theorem}
\begin{proof}
    Let there be the definition of regret $R_n$ which can be expressed as follows:
    \begin{align}
        R_n & = \sum_{t=1}^n \mathbb{E}\left[X_t^* - X_t\right] \nonumber \\
        & = \sum_{t=1}^n \mathbb{E}\left[\left\langle a_t^* - a_t , z_t \right\rangle \mid \left\langle a - a', z_t \right\rangle \geq 0, a_t^* = a\right] \nonumber, 
    \end{align}
    \begin{multline}\label{eq:regret_breakdown}
        \Rightarrow R_n \overset{(a)}{=} \\ \sum_{t=1}^n \sum_{a,a' \in \mathcal{A}} \mathbb{E}_{z_t}\left[\left\langle a_t^* - a_t , z_t \right\rangle \mid a_t^*, a_t\right]P\left(a_t = a'\mid a_t^* = a\right), 
    \end{multline}
    where in $(a)$ we used the Law of Total of Expectation. Let us assume at round $t$ that the action selected by the \textit{Kalman Oracle} is $a \in \mathcal{A}$. We want to find the probability that the \textit{Kalman Oracle} \eqref{eq:LGDS_Oracle_Kalman_Filter} chooses an action $a' \in \mathcal{A}$ such that $a' \neq a$. The event of this occurring is based on the following sets
    \begin{align}
        \mathcal{E}_t^a & \triangleq \cap_{a' \in \mathcal{A}} \left\{\left\langle a - a', z_t\right\rangle > 0\right\} \nonumber \\
        \Tilde{\mathcal{E}}_t^a & \triangleq \cap_{a' \in \mathcal{A}} \left\{\left\langle a - a', \Tilde{z}_{t|t-1}\right\rangle > 0\right\} \nonumber . 
    \end{align}

    Next, we want to find the distribution of $\left(\left\langle a - a', \Tilde{z}_{t|t-1}\right\rangle,\left\langle a - a', z_t\right\rangle\right)$. Recall that in the Kalman filter the state prediction $z_t = \Tilde{z}_{t|t-1}+\Tilde{e}_{t|t-1}$. Therefore the joint distribution of $\left(\left\langle a - a', \Tilde{z}_{t|t-1}\right\rangle,\left\langle a - a', z_t\right\rangle\right)$ is
    \begin{align}
        & P\left(\mathcal{E}_t^{a_j}\mid \Tilde{\mathcal{E}}_t^{a_i} \right) \nonumber \\
        & ~~ = \int_{\mathbb{R}_{+}^{k-1}}\int_{\mathbb{R}_{+}^{k-1}} P\left(A_i \Tilde{z}_{t|t-1} + \Pi_{i|j} \zeta = \Tilde{\zeta} \mid A_j z_t = \zeta\right)  d\zeta d\Tilde{\zeta} \nonumber \\
        & ~~ = \int_{\mathbb{R}_{+}^{k-1}}\int_{\mathbb{R}_{+}^{k-1}} \frac{\exp\left(-\frac{\left(\Tilde{\zeta} - \Pi_{i|j} \zeta\right)^\top \Tilde{\Sigma}_{i|j}^{-1}\left(\Tilde{\zeta} - \Pi_{i|j} \zeta \right)  }{2}\right)}{\sqrt{\left(2\pi\right)^{2k-2} \left\vert \Tilde{\Sigma}_{i|j}\right\vert }} d\zeta d\Tilde{\zeta} \nonumber \\
        & ~~ \overset{(b)}{=} \int_{\mathbb{R}_{+}^{2k-2}}\frac{\exp\left(-\frac{1}{2}\Vec{\zeta}^\top \Psi_{i|j}   \Vec{\zeta}\right)}{\sqrt{\left(2\pi\right)^{2k-2} \left\vert \Tilde{\Sigma}_{i|j}\right\vert }} d\Vec{\zeta} \nonumber \\
        & ~~ = \int_{\mathbb{R}_{+}^{2k-2}}\frac{\exp\left(-\frac{1}{2}\mbox{tr}\left(\Vec{\zeta}\Vec{\zeta}^\top \Psi_{i|j}   \right)\right)}{\sqrt{\left(2\pi\right)^{2k-2} \left\vert \Tilde{\Sigma}_{i|j}\right\vert }} d\Vec{\zeta} \nonumber \\
        & ~~ \geq \int_{\mathbb{R}_{+}^{2k-2}}\frac{\exp\left(-\frac{1}{2}\mbox{tr}\left(\Vec{\zeta}\Vec{\zeta}^\top \right) \mbox{tr}\left(\Psi_{i|j}\right)   \right)}{\sqrt{\left(2\pi\right)^{2k-2} \left\vert \Tilde{\Sigma}_{i|j}\right\vert }} d\Vec{\zeta} \nonumber \\
        & ~~ = \frac{\prod_{s=1}^{2k-2}\int_0^\infty \exp\left(-\frac{\Vec{\zeta}[s]^2}{2\mbox{tr}\left(\Psi_{i|j}\right)^{-1}} \right) d\Vec{\zeta}[s]}{\sqrt{\left(2\pi\right)^{2k-2} \left\vert \Tilde{\Sigma}_{i|j}\right\vert }}  \nonumber \\
        & ~~ = \frac{\prod_{s=1}^{2k-2} \sqrt{\frac{2\pi}{\mbox{tr}\left(\Psi_{i|j}\right)}}}{\sqrt{\left(2\pi\right)^{2k-2} \left\vert \Tilde{\Sigma}_{i|j}\right\vert }}  \nonumber, 
    \end{align}
    \begin{equation}
        \Rightarrow P\left(\mathcal{E}_t^{a_j}\mid \Tilde{\mathcal{E}}_t^{a_i} \right) \geq \frac{1}{\sqrt{\mbox{tr}\left(\Psi_{i|j}\right)^{2k-2} \left\vert \Tilde{\Sigma}_{i|j}\right\vert }} \nonumber ,
    \end{equation}
    where in $(b)$ we replaced $\tilde{\Sigma}_{i|j}$ with $\Psi_{i|j}$. Finally, we need the expectation $\mathbb{E}_{z_t}\left[\left\langle a_t^* - a_t , z_t \right\rangle \mid a_t^*, a_t\right]$. We know that based on the definition of $a_t^*$, $\left\langle a_t^* - a_t , z_t \right\rangle > 0$. We also know that $z_t$ is a normally distributed random variable $z_t \sim \mathcal{N}\left(\mathbf{0},Z\right)$ where $Z = \Gamma Z \Gamma^\top + Q$. Therefore, the conditional expectation is 
    \begin{equation}
        \mathbb{E}_{z_t}\left[\left\langle a_t^* - a_t , z_t \right\rangle \mid a_t^*, a_t\right] = \sqrt{\frac{2\left(a_t^* - a_t \right)^\top Z \left(a_t^* - a_t \right)}{\pi}}\nonumber. 
    \end{equation}

    Therefore, regret for the \textit{Kalman Oracle} is \eqref{eq:regret_lower_bound_discrete}. 

\end{proof}

Theorems \ref{theorem:lower_bound} and \ref{theorem:lower_bound_discrete} state directly that any policy must have at least a linearly increasing regret rate. The rationale is that the accumulation of the errors increases linearly, which implies that for any round $t$ the policy will choose the suboptimal action with a high probability. However, it is possible to still get a regret that is almost zero if (1) the lower bound error covariance matrix $P' = \mathbf{0}$ for the continuous action-space case or (2) $a_j - a_i$ is always in the null-space of $Z$ found in \eqref{eq:regret_lower_bound} for the discrete action-space case. The same is true for the edge case $k = 1$, since the learner can only select the optimal action. 

The results of this section and Section \ref{sec:difficulty} imply optimality is computationally intractable to obtain and the optimal policy does not guarantee consistent optimal action selection. Therefore, we will propose in the next section to select actions that maximize the reward prediction perturbed by a value. We motivate this strategy as it will be proven that these methods increase linearly similarly to the lower regret bound.

\section{Adding a Perturbation Value}\label{sec:Optimism-Based Methods}

Based on the results of Sections \ref{sec:Lower Bound of Environments modeled as LGDS}, we analyzed that regret is always linearly increasing with respect to error $P_t$ and the state prediction $\hat{z}_t$. Therefore, we propose to analyze algorithms of the following form
\begin{equation}\label{eq:optimistic_action_selection}
    a_t = \underset{a \in \mathcal{A}}{\arg\max} \left\langle a, \hat{z}_{t|t-1}\right\rangle + u_t\left(a\mid P_{t|t-1}\right),
\end{equation}
where $u_t\left(a_t\mid P_{t|t-1}\right) \in \mathbb{R}$, $a \in \mathcal{A}$, is denoted as the optimism term. Actions selected based on \eqref{eq:optimistic_action_selection} can be interpreted as a trade-off between choosing actions that the learner predicts to return the highest reward (i.e. $\arg\max_{a \in \mathcal{A}} \left\langle a, \hat{z}_{t|t-1}\right\rangle$) versus choosing actions based on $u_t\left(a\mid P_{t|t-1}\right)$. The following theorem proves that policies that select actions based on \eqref{eq:optimistic_action_selection} have an regret upper bound that increases linearly, similar to the lower bound in \eqref{eq:regret_lower_bound}. 

\begin{theorem}\label{theorem:regret_bound}
    Let $a_t \in \mathcal{A}$ be the learner's chosen action that returns reward $X_t$ at round $t$. In addition, let $a_t^* \in \mathcal{A}$ be the action that returns highest reward $X_t^*$ at round $t$. For actions selected based on \eqref{eq:optimistic_action_selection}, the upper bound for regret $R_n$ is 
    \begin{multline}\label{eq:Optimistic_Regret_Bound}
        R_n \leq \sum_{t = 1}^n u_t\left(a_t\mid P_{t|t-1}\right) - u_t\left(a_t^*\mid P_{t|t-1}\right) \\ + 2 \left\Vert e_{t|t-1}\right\Vert_2 . 
    \end{multline}
    
    Since $\left\Vert e_{t|t-1}\right\Vert_2 \geq 0$ almost surely occurs, then the upper bound on regret for policies that select actions based on \eqref{eq:optimistic_action_selection} increases at least linearly. 
\end{theorem}

\begin{proof}
    Since $z_t = \hat{z}_{t|t-1} + e_{t|t-1}$ where $\hat{z}_{t|t-1}$ is the Kalman filter state prediction and $e_{t|t-1}$ is the error of the state prediction, we can add and subtract $u_t\left(a_t^*\mid P_{t|t-1}\right)$ to instantaneous regret $r_t$ to provide the following expression of $r_t$:
    \begin{multline}
        r_t = \left\langle a_t^*, \hat{z}_{t|t-1} \right\rangle + u_t\left(a_t^*\mid P_{t|t-1}\right) + \left\langle a_t^*,e_{t|t-1} \right\rangle \\ - \left\langle a_t, \hat{z}_{t|t-1} + e_{t|t-1} \right\rangle  - u_t\left(a_t^*\mid P_{t|t-1}\right). \nonumber
    \end{multline}
   
    Since the learner chooses action $a_t \in \mathcal{A}$ at round $t$, then $\left\langle a_t^*, \hat{z}_{t|t-1} \right\rangle + u_t\left(a_t^*\mid P_{t|t-1}\right)$ can be upper bounded as follows:
    \begin{multline}\label{eq:optimistic_action_selection_bound}
        \left\langle a_t^*, \hat{z}_{t|t-1} \right\rangle + u_t\left(a_t^*\mid P_{t|t-1}\right) \leq \left\langle a_t, \hat{z}_{t|t-1} \right\rangle \\ + u_t\left(a_t\mid P_{t|t-1}\right).  
    \end{multline}

    Using inequality \eqref{eq:optimistic_action_selection_bound}, regret has upper bound
    \begin{multline}\label{eq:regret_upper_bound}
        r_t \leq u_t\left(a_t\mid P_{t|t-1}\right) - u_t\left(a_t^*\mid P_{t|t-1}\right) \\ + \left\langle a_t^* - a_t, e_{t|t-1} \right\rangle . 
    \end{multline}

    Finally, since $a_t,a_t^* \in \mathcal{A}$ has norm 1, i.e. $\left\Vert a \right\Vert_2 = 1$ for $a \in \mathcal{A}$, then we can upper bound \eqref{eq:regret_upper_bound} as \begin{multline}\label{eq:regret_error_bound}
        r_t \leq u_t\left(a_t\mid P_{t|t-1}\right) - u_t\left(a_t^*\mid P_{t|t-1}\right) \\ + 2\left\Vert e_{t|t-1}\right\Vert_2 . 
    \end{multline} 

    Therefore, the upper-bound on regret $R_n$ \eqref{eq:regret} is \eqref{eq:Optimistic_Regret_Bound}. 
\end{proof}

In Theorem \ref{theorem:regret_bound}, the inequality \eqref{eq:regret_error_bound} is based only on \eqref{eq:optimistic_action_selection_bound} and the norm of each action $a \in \mathcal{A}$, which is 1. Next, since instantaneous regret $r_t$ is always nonnegative, i.e. $r_t \geq 0$, then according to inequality \eqref{eq:regret_upper_bound} of Theorem \ref{theorem:regret_bound}, if we restrict the design of $u_t\left(a\right) \geq 0$ for $a \in \mathcal{A}$, the following inequality is always satisfied:
\begin{multline}
    \left\langle a_t^*, e_{t|t-1} \right\rangle - u_t\left(a_t^*\mid P_{t|t-1}\right)\geq \\ \left\langle a_t, e_{t|t-1} \right\rangle - u_t\left(a_t\mid P_{t|t-1}\right) \nonumber. 
\end{multline}

Theorem \ref{theorem:regret_bound} implies that if the LGDS \eqref{eq:LGDS} has a stable state matrix $\Gamma$, then the difference between the bound \eqref{eq:Optimistic_Regret_Bound} and Theorem \ref{theorem:lower_bound}'s bound \eqref{eq:regret_lower_bound} is constant. This constant is impacted directly by the magnitude of the optimism term $u_t\left(a_t\mid P_{t|t-1}\right)$ and the error $e_{t|t-1}$. Based on above, if $\left\langle a_t^*, e_{t|t-1} \right\rangle \geq \left\langle a_t, e_{t|t-1} \right\rangle$, then $u_t\left(a_t\mid P_{t|t-1}\right)$ is too large. However, we want $u_t\left(a_t\mid P_{t|t-1}\right)$ to be as close as possible to the magnitude of $\left\langle a_t, e_{t|t-1} \right\rangle$ to lower the upper bound of regret in \eqref{eq:regret_upper_bound}. Therefore, we propose two methods: Kalman filter Upper Confidence Bound (Kalman-UCB) (Algorithm \ref{alg:Kalman-UCB}) and Information filter Directed Exploration for Action-selection (IDEA) (Algorithm \ref{alg:IDEA}). 

In each of the algorithms, there exists the steps \verb|Action Selection|, \verb|Observation|, and \verb|Update|. In each method's \verb|Action Selection|, the learner selects the action with the highest reward prediction perturbed by value, which we will review in the following subsections. For \verb|Observation|, the learner observes the reward $X_t$ which is based on the learner's selected action $a_t$. Finally, in \verb|Update|, the learner updates the Kalman filter posed in \eqref{eq:Kalman_Filter}.

\begin{algorithm}[!t]
\caption{Kalman filter Upper Confidence Bound (Kalman-UCB)}\label{alg:Kalman-UCB}
 \begin{algorithmic}[1]
\STATE \textbf{Input}: $\Gamma$, $\mathcal{A}$, $Q$, $\sigma$, $\Sigma_0$, $z_0$
\FOR{$t=1,2,\dots,n$}
    \STATE \verb|/* Action Selection */|
    \STATE $a_t = \underset{a \in \mathcal{A}}{\arg\max} \left\langle a, \hat{z}_{t|t-1} \right\rangle + \sqrt{a^\top P_{t|t-1}a}$  
    \STATE \verb|/* Observation */|
    \STATE Observe $X_t = \left\langle a_t, z_t \right\rangle + \eta_t$
    \STATE \verb|/* Update */|
    \STATE Update $\hat{z}_{t+1|t}$ and $P_{t+1|t}$ in the Kalman filter \eqref{eq:Kalman_Filter}
\ENDFOR
\end{algorithmic}
\end{algorithm}

\begin{algorithm}[!t]
\caption{Information filter Directed Exploration for Action-selection (IDEA)}\label{alg:IDEA}
 \begin{algorithmic}[1]
\STATE \textbf{Input}: $\Gamma$, $\mathcal{A}$, $Q$, $\sigma$, $\Sigma_0$, $z_0$
\FOR{$t=1,2,\dots,n$}
    \STATE \verb|/* Action Selection */|
    \STATE $a_t = \underset{a \in \mathcal{A}}{\arg\max} \left\langle a, \hat{z}_{t|t-1} \right\rangle + \sqrt{\mbox{tr}\left(\frac{\Gamma P_{t|t-1} a a^\top P_{t|t-1}\Gamma^\top}{a^\top P_{t|t-1}a + \sigma^2}\right)}$  
    \STATE \verb|/* Observation */|
    \STATE Observe $X_t = \left\langle a_t, z_t \right\rangle + \eta_t$ 
    \STATE \verb|/* Update */|
    \STATE Update $\hat{z}_{t+1|t}$ and $P_{t+1|t}$ in the Kalman filter \eqref{eq:Kalman_Filter}
\ENDFOR
\end{algorithmic}
\end{algorithm}

\subsection{Optimism in the Face of Uncertainty: Kalman-UCB (Algorithm \ref{alg:Kalman-UCB})}\label{sec:UCB_Method}

Kalman-UCB is based on a principle commonly used for SMAB: optimism in the face of uncertainty. Therefore, Kalman-UCB's perturbation is based on the upper confidence bound on the reward prediction $\left\langle a, \hat{z}_{t|t-1}\right\rangle$, i.e. with a probability of at least $1-\delta$, where $\delta, \in (0,1)$
\begin{equation}
    \left\vert X_t-\left\langle a, \hat{z}_{t|t-1}\right\rangle\right\vert \leq \sqrt{\left(a^\top P_{t|t-1} a+\sigma^2\right)\log\left(1/\delta\right)} \nonumber. 
\end{equation}

Therefore, Kalman-UCB selects actions based on the following optimization problem 
\begin{multline}\label{eq:UCB_Method}
    a_{t+1} = \underset{a \in \mathcal{A}}{\arg\max} \left\langle a, \hat{z}_{t|t-1}\right\rangle \\ + \sqrt{\left(a^\top P_{t|t-1} a\right)\log\left(1/\delta\right)}.
\end{multline}
where $\sigma^2\log\left(1/\delta\right)$ is removed since it is independent of the action $a \in \mathcal{A}$. To study Kalman-UCB's exploration behavior, we will focus on the sequence of actions that only maximize the perturbation $\sqrt{a^\top P_{t|t-1}a}$. The following lemma is provided for theoretical insight. 

\begin{lemma}\label{lemma:periodic_seq}
    Let $P_a$, $a \in \mathcal{A}$, be the solution of the Algebraic Riccati Equation (ARE), i.e. $P_a = g\left(P_a,a\right)$ where $g\left(P_a,a\right)$ is defined in \eqref{eq:g_definition}. If for every action $a \in \mathcal{A}$ there exists another action $a' \in \mathcal{A}$ such that $\sqrt{a^\top P_a a} \leq \sqrt{\left(a'\right)^\top P_a a'}$, every action $a \in \mathcal{A}$ will be sampled periodically. 
\end{lemma}

\begin{proof}
    For every action $a \in \mathcal{A}$, the covariance matrix $P_{t|t-1}$ converges exponentially to $P_a$ as $t$ increases, where $P_a$ is the solution of the ARE $P_a = g\left(P_a,a\right)$ where $g\left(P_a,a\right)$ is defined in \eqref{eq:g_definition}. Since for every action $a \in \mathcal{A}$ there exists another action $a' \in \mathcal{A}$, $a' \neq a$, such that $\sqrt{a^\top P_a a} \leq \sqrt{\left(a'\right)^\top P_a a'}$, then 
    \begin{equation}
        a' = \underset{a \in \mathcal{A}}{\arg\max} ~ \sqrt{a^\top P_{t|t-1}a} \nonumber. 
    \end{equation}

    Since this happens for every action $a \in \mathcal{A}$ and $g\left(P_{t|t-1},a\right)$ is deterministic, then $\sqrt{a^\top P_{t|t-1}a}$ is periodic. 
\end{proof}

Lemma \ref{lemma:periodic_seq} states that if there exists two actions $a,a' \in \mathcal{A}$ such that $a^\top P_a a \leq \left(a'\right)^\top P_a a'$ and $\left(a'\right)^\top P_{a'} a' \leq a^\top P_{a'} a$ where $P_a$ and $P_{a'}$ are the stable error covariance matrices of actions $a$ and $a'$, respectively, then the sequence $\left\{\arg\max_{a \in \mathcal{A}} u_t\left(a \mid P_{t|t-1}\right)\right\}_{t=1}^n$ will switch between actions $a$ and $a'$ for $t = 1,2,\dots,n$. This implies that Kalman-UCB has an implicit periodic schedule of actions that it explores. 

Since $P_a$ has different magnitudes for different actions $a\in \mathcal{A}$, this can lead to situations where an action $a' \in \mathcal{A}$ provides a lower $\mbox{tr}\left(P_{a'}\right)$ even though action $a \in \mathcal{A}$ is selected since it maximizes $\sqrt{a^\top P_{t|t-1} a}$. In effect, action $a' \in \mathcal{A}$ lowers the prediction error $\sqrt{a^\top P_{t|t-1} a + \sigma^2}$ for all actions $a \in\mathcal{A}$, implying that selecting this action is more beneficial than lowering each action's error individually. Therefore, the next section will address this perspective.

% Therefore, we now have to consider how a sequence of actions $a_{t-s},\dots,a_t \in \mathcal{A}$ are observable instead of one action $a \in \mathcal{A}$ being observable (see \eqref{eq:observability_gramian}). First, according to \cite{mo2014infinite}, there exists a periodic sequence of actions $a_{t-s},\dots,a_t \in \mathcal{A}$ such that the error covariance matrix $P_{t|t-1}$ is bounded. Therefore, we propose to set $u_{t+1}\left(a\right)$ such that the action sequence $\left\{\arg\max_{a \in \mathcal{A}} u_t\left(a\right)\right\}_{t=1}^n$ is periodic. One such $u_t\left(a\right)$ that is a periodic sequence and is close to the magnitude of $\left\langle a, e_{t|t-1}\right\rangle$ (see Theorem \ref{theorem:regret_bound}) is 
% \begin{equation}\label{eq:UCB_u_version}
%     u_{t+1}\left(a\right) = \sqrt{a^\top P_{t+1|t} a}. 
% \end{equation}

% The theorem below proves that if $u_{t+1}\left(a\right)$ is set to \eqref{eq:UCB_u_version}, then there exists an upper bound $\overline{p}$ such that $\sqrt{a^\top P_{t+1|t} a} \leq \overline{p}$ for any $a \in \mathcal{A}$. 

\subsection{Using Observability: IDEA (Algorithm \ref{alg:IDEA})}\label{sec:Observer Method}

% In Section \ref{sec:Myopic}, we reviewed the \textit{Greedy Method} and its performance. The performance of the \textit{Greedy Method} was proven to be a function of the largest steady-state error covariance matrix $P_{\overline{a}}$ associated with action $\overline{a} \in \mathcal{A}$. In this section, we will visit discuss IDEA. Recall that the \textit{Greedy Method} is the solution of a 1-step dynamic programming optimization problem. 

IDEA aims to address the perspective presented in Kalman-UCB: if an action $a' \in \mathcal{A}$ lowers the prediction error for all actions more effectively than each action $a \in \mathcal{A}$ individually, why not explore the LGDS environment by selecting that action $a' \in \mathcal{A}$ repeatedly? To implement this idea, we will approximate the two-step dynamic programming where the continuous set of actions constrained to the unit sphere is used. 

\begin{theorem}\label{theorem:approximation}
    Let there be IDEA which optimization problem \eqref{eq:Observer_Method}. There exists an optimization problem that bounds the 2-step dynamic programming optimization where actions are on the unit sphere
    \begin{multline}\label{eq:upper_bound_triangle}
        V_{n-1}\left(z_{n-1}\right) \leq \max_{a \in \mathcal{A}} \left\langle a, \hat{z}_{n-1|n-2} \right\rangle + 
        \left\Vert \Gamma \hat{z}_{n-1|n-2}\right\Vert_2 \\
        + \sqrt{\mbox{tr}\left(\frac{\Gamma P_{n-1|n-2} aa^\top P_{n-1|n-2} \Gamma^\top}{a^\top P_{n-1|n-2} a + \sigma^2}\right)}. 
    \end{multline}
    % \begin{multline}\label{eq:upper_bound_jensen}
    %     V_{n-1}\left(z_{n-1}\right) \leq \max_{a \in \mathcal{A}} \left\langle a, \hat{z}_{n-1|n-2} \right\rangle + \\ 
    %     \sqrt{\left\Vert \Gamma \hat{z}_{n-1|n-2}\right\Vert_2^2 +
    %     \mbox{tr}\left(\frac{\Gamma P_{n-1|n-2} aa^\top P_{n-1|n-2} \Gamma^\top}{a^\top P_{n-1|n-2} a + \sigma^2}\right)}.
    % \end{multline}
\end{theorem}

\begin{proof}
    Recall in Theorem \ref{theorem:Separation_Principle} that $V_{n-1}\left(z_{n-1}\right)$ is expressed as \eqref{eq:n-1_step}. Using Lemma \ref{lemma:Kalman_Facts}, we can express $V_{n-1}\left(z_{n-1}\right)$ as follows:
    \begin{multline}
        V_{n-1}\left(z_{n-1}\right) = \max_{a \in \mathcal{A}} \left\langle a, \hat{z}_{n-1|n-2}\right\rangle + \\
        \mathbb{E}\left[\sqrt{\left\Vert \hat{z}_{n|n-1}\right\Vert_2^2 + \mbox{tr}\left(P_{n|n-1} - g\left(P_{n-1|n-2}, a\right)\right)} \mid \mathcal{F}_{n-2}\right] \nonumber,
    \end{multline}
    \begin{multline}
        V_{n-1}\left(z_{n-1}\right) = \max_{a \in \mathcal{A}} \left\langle a, \hat{z}_{n-1|n-2}\right\rangle \\
        +\mathbb{E}\left[\sqrt{\left\Vert \hat{z}_{n|n-1}\right\Vert_2^2} \mid \mathcal{F}_{n-2}\right] \label{eq:dynamic_iteration}.
    \end{multline}
    
    The state prediction $\hat{z}_{n|n-1} \in \mathbb{R}^d$ can be expressed by the following Kalman filter iteration
    \begin{equation}\label{eq:Markov_Formulation}
        \begin{cases}
            \hat{z}_{t+1|t} & = \Gamma \hat{z}_{t|t-1} \\
            & ~~~ + \Gamma P_{t|t-1} a_t \left(a_t^\top P_{t|t-1} a_t + \sigma^2\right)^{-1/2} \omega_t \\
            X_t & = \left\langle a_t, \hat{z}_{t|t-1} \right\rangle + \left(a_t^\top P_{t|t-1} a_t + \sigma^2\right)^{1/2} \omega_t
        \end{cases}, 
    \end{equation}
    where $\omega_t \in \mathbb{R}$ is from the standard normal distribution, i.e. $\omega_t \sim \mathcal{N}\left(0,1\right)$. Using \eqref{eq:Markov_Formulation} we can express \eqref{eq:dynamic_iteration} as 
    \begin{multline}
        V_{n-1}\left(z_{n-1}\right)= \max_{a \in \mathcal{A}} \left\langle a, \hat{z}_{n-1|n-2} \right\rangle 
        \\
        +\mathbb{E}\left[\left\Vert \Gamma \hat{z}_{n-1|n-2} + \frac{\Gamma P_{n-1|n-2} a \omega_{n-1}}{\sqrt{a^\top P_{n-1|n-2}a + \sigma^2 }} \right\Vert_2 \mid \mathcal{F}_{n-2}\right]. \label{eq:problem_with_2-step} 
    \end{multline}

    Using the the triangle inequality \eqref{eq:problem_with_2-step} provides \eqref{eq:upper_bound_triangle}. Finally, the optimization problem in \eqref{eq:upper_bound_triangle} is equivalent to IDEA's action selection strategy as the chosen actions are independent of the norm $\left\Vert \Gamma \hat{z}_{t|t-1}\right\Vert_2$. 

\end{proof}

As shown in Theorem \ref{theorem:approximation}, we approximate the $n-1$ step of the dynamic programming problem with \eqref{eq:upper_bound_triangle}. This approximation introduces a perturbation value that is the $\ell_2$ norm of the find matrix product in $g\left(P_{t|t-1},a_t\right)$ defined in \eqref{eq:g_definition}. This final term is the amount of the error $\Gamma P_{t|t-1} \Gamma^\top 
+Q$ decreases from the feedback $\Gamma K_t \left(X_t-\left\langle a_t,\hat{z}_{t|t-1}\right\rangle\right)$. Therefore, by choosing actions that maximize \eqref{eq:Observer_Method} or \eqref{eq:upper_bound_triangle}, we are balancing between choosing the action that maximizes predicted reward $\left\langle a_t,\hat{z}_{t|t-1}\right\rangle$ versus the action that maximizes the amount of feedback $\Gamma K_t \left(X_t-\left\langle a_t,\hat{z}_{t|t-1}\right\rangle\right)$. Therefore, IDEA selects actions based on the following optimization problem 
\begin{multline}\label{eq:Observer_Method}
    a_t = \underset{a \in \mathcal{A}}{\arg\max} \left\langle a, \hat{z}_{t|t-1}\right\rangle \\ + \sqrt{\mbox{tr}\left(\frac{\Gamma P_{t|t-1} aa^\top P_{t|t-1} \Gamma^\top}{a^\top P_{t|t-1} a + \sigma^2}\right)}.
\end{multline}

\section{Discussion on Kalman-UCB and IDEA Exploration Methodologies}\label{sec:Discussion}

Kalman-UCB and IDEA exploration methodologies are fairly different. Kalman-UCB explores actions with the highest reward prediction error. This can be advantageous if LGDS \eqref{eq:LGDS} lacks an observable action $a \in \mathcal{A}$. IDEA explores by choosing the action that maximizes the feedback error term in the Kalman filter \eqref{eq:Kalman_Filter}. In effect, IDEA minimizes the predicted LGDS state variable error. This is beneficial if there exists an action that minimizes the reward prediction error for all other actions. The next section provides an analysis for comparing the performance of Kalman-UCB and IDEA, where performance will be based on accuracy of selecting the \textit{Oracle}'s action. 

\subsection{Metric of Performance}\label{sec:performance_comparison_theory}

To provide a metric for comparing the performance of Kalman-UCB and IDEA, we first provide the following Lemma \ref{lemma:probability_wrong_action}. Using Lemma \ref{lemma:probability_wrong_action}, we then provide an interval of performance for Kalman-UCB and IDEA, which can compared between the two methods to measure which method will perform better. 

\begin{lemma}\label{lemma:probability_wrong_action}
    Let us assume that the error covariance matrix for each method is equivalent, i.e. $P_{t|t-1} \equiv P$. Also, let $\mu_i \in \mathbb{R}^{2(k-1)}$ and $\hat{\Sigma}_{i,j} \in \mathbb{R}^{2(k-1)\times 2(k-1)}$ be defined to be the vector and matrix 
    \begin{align}
        \mu_i\left(P\right) & \triangleq \begin{pmatrix}
            u_t\left(a_i \mid P \right) - u_t \left(a_1 \mid P\right) \\
            \vdots \\
            u_t\left(a_i \mid P \right) - u_t \left(a_{k-1} \mid P\right) \\
            \mathbf{0}_{k-1}
        \end{pmatrix} \\
        \hat{\Sigma}_{i,j}\left(P\right) & \triangleq \begin{pmatrix}
            A_i \left(Z - P\right) A_i^\top & A_i\left(Z - P\right) A_j^\top \\
            A_j \left(Z - P\right) A_i^\top & A_j Z A_j^\top 
        \end{pmatrix},
    \end{align}
    where $u_t\left(a_i \mid P\right)$ is the perturbation added in an optimism-based method. The probability that an optimism-based chooses an action not equal to the \textit{Oracle}'s action $a$ is 
    \begin{multline}
        P\left(\hat{\mathcal{U}}_t^{a_i} \mid \mathcal{U}_t^{a_j}\right) = \\ \frac{\int_{\mathbb{R}_{+}^{k-1}}\int_{\mathbb{R}_{+}^{k-1}} P\left(A_i \hat{z}_{t|t-1} + \Delta u_i = \hat{\zeta}, A_j z_t = \zeta\right)  d\zeta d\hat{\zeta}}{\int_{\mathbb{R}_{+}^{k-1}}\int_{\mathbb{R}^{k-1}} P\left(A_i \hat{z}_{t|t-1} + \Delta u_i = \hat{\zeta} , A_j z_t = \zeta\right)  d\zeta d\hat{\zeta}} \label{eq:probability_wrong_action}, 
    \end{multline}
    where the distribution in the integral is defined as 
    \begin{multline}
        P\left(A_i \hat{z}_{t|t-1} + \Delta u_i = \hat{\zeta} , A_j z_t = \zeta\right) \\ = \mathcal{N}\left(\mu_i\left(P\right),\hat{\Sigma}_{i,j}\left(P\right)\right) \label{eq:distribution_optimism_action}. 
    \end{multline}
\end{lemma}

\begin{proof}
    We want to find the probability that an optimism-based method chooses an $a_i \in \mathcal{A}$ such that $a_i \neq a_j$. The event of this occurring is based on the following sets
    \begin{align}
        \mathcal{U}_t^{a_j} & \triangleq \cap_{a' \in \mathcal{A}} \left\{\left\langle a_j - a' , z_t\right\rangle > 0\right\} \\
        \hat{\mathcal{U}}_t^{a_i} & \triangleq \cap_{a' \in \mathcal{A}} \left\{\left\langle 
        a_i - a', \hat{z}_{t|t-1}\right\rangle + \Delta u_i > 0\right\} \nonumber \\
        & = \cap_{a' \in \mathcal{A}} \left\{\left\langle a_i
         - a', z_t - e_{t|t-1}\right\rangle + \Delta u_i > 0\right\}
    \end{align}
    where $\Delta u_i \triangleq u_t\left(a_i\mid P\right) - u_t\left(a'\mid P\right)$. We want to compute the distribution of the event $\mathcal{U}_t^{a_i} \mid \hat{\mathcal{U}}_t^{a_j}$ as follows:
    \begin{multline}
        P\left(\hat{\mathcal{U}}_t^{a_i} \mid \mathcal{U}_t^{a_j}\right) =\\  \int_{\mathbb{R}_{+}^{k-1}}\int_{\mathbb{R}_{+}^{k-1}} P\left(A_i \hat{z}_{t|t-1} + \Delta u_i = \hat{\zeta} \mid A_j z_t = \zeta\right)  d\zeta d\hat{\zeta} \nonumber, 
    \end{multline}
    leading to \eqref{eq:probability_wrong_action}. The distribution in the integral is defined as \eqref{eq:distribution_optimism_action}. 
\end{proof}

The only difference in expected regret for any optimism-based method is $\mu_i$ in \eqref{eq:distribution_optimism_action}. Therefore, instead directly measuring regret as a metric for comparing performances between each optimism-based method, we will instead analyze the Wasserstein metric between two distributions, where the first distribution will be the distribution is \eqref{eq:distribution_optimism_action}, while the second distribution is the distribution $\mathcal{N}\left(\mathbf{0},\Sigma_{i,j}\right)$ where $\Sigma_{i,j}$ is defined as \eqref{eq:oracle_covariance}. 
\begin{multline}\label{eq:metric_comparison}
    \phi\left(i,j \mid P\right) = \left\Vert \mu_i \right\Vert_2 + \mbox{tr}\left(\Sigma_{i,j} +\hat{\Sigma}_{i,j}\left(P\right) \right) \\ - 2\mbox{tr}\left(\left(\hat{\Sigma}_{i,j}\left(P\right)^{1/2} \Sigma_{i,j}\hat{\Sigma}_{i,j}\left(P\right)^{1/2}\right)^{1/2} \right). 
\end{multline}

The interpretation of this metric \eqref{eq:metric_comparison} centers on the following question: Given the distribution of the LGDS state variable $z_t$, to what extent does the perturbation signal $u_t\left(a_t \mid P_{t|t-1}\right)$ impact the reward prediction $\left\langle a, \hat{z}_{t|t-1}\right\rangle$ such that the learner selects the suboptimal action? Consequently, this measure implies that if the perturbation $u_t\left(a\mid P_{t|t-1}\right)$ is small, then the method that uses $u_t\left(a\mid P_{t|t-1}\right)$ will have better performance. We utilize the metric \eqref{eq:metric_comparison} to compare the performance between Kalman-UCB and IDEA with the interval
\begin{equation}\label{eq:interval_performance}
    \begin{pmatrix}
        \min_{i \neq j, a \in \mathcal{A}} \phi\left(i,j \mid P_a\right), \max_{i \neq j ,a \in \mathcal{A}} \phi\left(i,j \mid P_a\right)
    \end{pmatrix}, 
\end{equation}
where $P_a$ represents the steady-state error covariance matrix of the Kalman filter error, which solves the ARE $P_a = g\left(P_a,a\right)$. 

The bounds of performance \eqref{eq:interval_performance} measures the influence of the optimism term $u_t\left(a_t \mid P_{t|t-1}\right)$ on the reward prediction $\left\langle a, \hat{z}_{t|t-1}\right\rangle$. A significant impact implies that the corresponding method will perform worse, while a minor impact indicates better performance. By using an interval with the bounds defined as smallest and largest $\phi\left(i,j \mid P_a\right)$ values, the impact of $u_t\left(a_t \mid P_{t|t-1}\right)$ can be studied for any initialized $P_{0|-1}$.  

\subsection{Performance of other Bandit Algorithms}

There are a number of bandit algorithms that are applicable to our proposed bandit environment posed in \eqref{eq:LGDS}. A well-known method that has been discussed earlier in the introduction is the Upper Confidence Bound (UCB) proposed by Auer, Cesa-Bianchi, and Fischer in \cite{auer2002finite}. This has been extended to non-stationary environments through the Sliding-Window UCB (SW-UCB) proposed by Garivier and Moulines in \cite{garivier2008upper}. The UCB and SW-UCB algorithms are posed as Algorithms \ref{alg:UCB} and \ref{alg:SW-UCB}, respectively. To understand the performance of these algorithms with respect to our proposed environments, we will provide the regret upper bounds in the theorem below. 

\begin{algorithm}[!t]
\caption{Upper Confidence Bound (UCB) Algorithm}\label{alg:UCB}
 \begin{algorithmic}[1]
\STATE \textbf{Input}: $\delta \in (0,1), R$
\STATE \verb|/* Initialization */|
\FOR{$a \in \mathcal{A}$}
    \STATE $N_a \gets 0$ 
    \STATE $S_a \gets 0$ 
    \STATE $\hat{\mu}_a \gets 0$
\ENDFOR
\FOR{$t=1,2,\dots,n$}
    \STATE \verb|/* Action Selection */|
    \STATE $a_t = \underset{a \in \mathcal{A}}{\arg\max} ~ \hat{\mu}_a + \sqrt{\frac{2R^2\log\left(1/\delta\right)}{N_a}}$  
    \STATE \verb|/* Observation */|
    \STATE Observe $X_t = \left\langle a_t, z_t \right\rangle + \eta_t$ 
    \STATE \verb|/* Update */|
    \STATE $N_{a_t} \gets N_{a_t} + 1$ 
    \STATE $S_{a_t} \gets S_{a_t} + X_t$ 
    \STATE $\hat{\mu}_{a_t} \gets \frac{S_{a_t}}{N_{a_t}}$
\ENDFOR
\end{algorithmic}
\end{algorithm}

\begin{algorithm}[!t]
\caption{Sliding Window UCB (SW-UCB) Algorithm}\label{alg:SW-UCB}
 \begin{algorithmic}[1]
\STATE \textbf{Input}: $\delta \in (0,1), R, T$
\STATE \verb|/* Initialization */|
\FOR{$a \in \mathcal{A}$}
    \STATE $\mathcal{T}_a \gets \{\}$ 
    \STATE $N_a \gets 0$ 
    \STATE $S_a \gets 0$ 
    \STATE $\hat{\mu}_a \gets 0$
\ENDFOR
\FOR{$t=1,2,\dots,n$}
    \STATE \verb|/* Action Selection */|
    \STATE $a_t = \underset{a \in \mathcal{A}}{\arg\max} ~ \hat{\mu}_a + \sqrt{\frac{2R^2\log\left(1/\delta\right)}{N_a}}$  
    \STATE \verb|/* Observation */|
    \STATE Observe $X_t = \left\langle a_t, z_t \right\rangle + \eta_t$ 
    \STATE \verb|/* Update */|
    \STATE $\mathcal{T}_{a_t} \gets \mathcal{T}_{a_t} \cup \{t\}$ 
    \FOR{$a \in \mathcal{A}$}
        \STATE $N_a \gets 0$ 
        \STATE $S_a \gets 0$ 
        \FOR{$\tau \in \mathcal{T}_a$}
            \IF{$\tau \in [t-T,t]$}
                \STATE $N_a \gets N_a + 1$ 
                \STATE $S_a \gets S_a + X_t$ 
                \STATE $\hat{\mu}_a \gets \frac{S_a}{N_a}$
            \ENDIF
        \ENDFOR
    \ENDFOR
\ENDFOR
\end{algorithmic}
\end{algorithm}

\begin{theorem}\label{theorem:regret_UCB}
    Let the reward $X_t$ be sampled from the SMAB environment \eqref{eq:LGDS}. UCB found in Algorithm \ref{alg:UCB} and SW-UCB found in Algorithm \ref{alg:SW-UCB} have the following regret upper bound which is satisfied with a probability of at least $1-\delta$ where $\delta \in (0,1)$: 
    \begin{equation}\label{eq:regret_UCB}
        \Rightarrow R_n^{UCB} \leq \max_{a \in \mathcal{A}} \sqrt{\left(3n^2 + n + 1\right)\left(a^\top Z_t a \log\left(1/\delta\right)\right)}.
    \end{equation}
    where $Z_t \triangleq \mathbb{E}\left[z_t z_t^\top\right]$ which is based on the iteration $Z_{t+1} = \Gamma Z_t \Gamma^\top + Q$. 
\end{theorem}

\begin{proof}
    For UCB's regret upper bound, we first bound the instantaneous regret $r_t^{UCB} \triangleq \left\langle a_t^*, z_t\right\rangle - \left\langle a_t, z_t\right\rangle$. The instantaneous regret $r_t^{UCB}$ for round $t$ using UCB can be expressed as 
    \begin{align}
        r_t^{UCB} & = \left\langle a_t^*, z_t\right\rangle - \left\langle a_t, z_t\right\rangle \nonumber \\
        & \overset{(a)}{\leq} \sqrt{2\left(a_t^*\right)^\top Z_t a_t^* \log\left(1/\delta\right)} - \left\langle a_t, z_t\right\rangle \nonumber \\
        & = \frac{\sqrt{N_{a_t^*}}}{\sqrt{N_{a_t^*}}}\left(\sqrt{2\left(a_t^*\right)^\top Z_t a_t^* \log\left(1/\delta\right)} - \left\langle a, z_t\right\rangle\right) \nonumber \\
        & = \sqrt{N_{a_t^*}}\sqrt{\frac{2\left(a_t^*\right)^\top Z_t a_t^* \log\left(1/\delta\right)}{N_a}} - \left\langle a_t, z_t\right\rangle \nonumber \\
        & = \sqrt{N_{a_t^*}}\left(\hat{\mu}_{a_t^*} + \sqrt{\frac{2\left(a_t^*\right)^\top Z_t a_t^* \log\left(1/\delta\right)}{N_{a_t^*}}}\right) \nonumber \\
        & ~~~~~~ - \left\langle a_t, z_t\right\rangle - \sqrt{N_{a_t^*}}\hat{\mu}_{a_t^*}\nonumber \\
        & \overset{(b)}{\leq} \sqrt{N_{a_t^*}}\left(\hat{\mu}_{a_t} + \sqrt{\frac{2a_t^\top Z_t a_t \log\left(1/\delta\right)}{N_{a_t}}}\right) \nonumber \\
        & ~~~~~~ - \left\langle a_t, z_t\right\rangle - \sqrt{N_{a_t^*}}\hat{\mu}_{a_t^*}\nonumber , 
    \end{align}
    \begin{multline*}
        \Rightarrow r_t^{UCB} \overset{(c)}{\leq} 3\sqrt{N_{a_t^*}}\sqrt{\frac{2a_t^\top Z_t a_t \log\left(1/\delta\right)}{N_{a_t}}} \\ + \sqrt{2a_t^\top Z_t a_t \log\left(1/\delta\right)} \nonumber. 
    \end{multline*}
    
    In $(a)$ we used the following inequality which is satisfied with a probability of at least $1-\delta$:
    \begin{equation*}
        \left\langle a_t^*, z_t\right\rangle \leq \sqrt{2\left(a_t^*\right)^\top Z_t a_t^* \log\left(1/\delta\right)} . 
    \end{equation*}
    
    In $(b)$ we used the following inequality: 
    \begin{multline*}
        \hat{\mu}_{a_t^*} + \sqrt{\frac{2\left(a_t^*\right)^\top Z_t a_t^* \log\left(1/\delta\right)}{N_{a_t^*}}} \leq \\ 
        \hat{\mu}_{a_t} + \sqrt{\frac{2a_t^\top Z_t a_t \log\left(1/\delta\right)}{N_{a_t}}}. 
    \end{multline*}        
    
    Finally, in $(c)$ we used the following inequality
    \begin{equation*}
        - \left\langle a_t, z_t\right\rangle \leq \sqrt{2a_t^\top Z_t a_t \log\left(1/\delta\right)}.
    \end{equation*}

    Note that regret is the sum of instantaneous regrets, i.e. $R_n = \sum_{t=1}^n r_t^{UCB}$: 
    \begin{equation*}
        R_n^{UCB} \leq \sum_{t=1}^n \left(3\sqrt{\frac{N_{a_t^*}}{N_{a_t}}} + 1\right)\sqrt{2a_t^\top Z_t a_t \log\left(1/\delta\right)} ,
    \end{equation*}
    \begin{multline*}
        \Rightarrow R_n^{UCB} \leq \\ \sqrt{\sum_{t=1}^n \left(3\frac{N_{a_t^*}}{N_{a_t}} + 6\sqrt{\frac{N_{a_t^*}}{N_{a_t}}} + 1\right)\left(2a_t^\top Z_t a_t \log\left(1/\delta\right)\right)},
    \end{multline*}
    leading to inequality \eqref{eq:regret_UCB} which is satisfied with a probability of at least $1-\delta$. 
    
\end{proof}

In Theorem \ref{theorem:regret_UCB}, the regret increases linearly with respect to the covariance of the LGDS state variable $z_t$. Based on the results of Theorem \ref{theorem:lower_bound_discrete}, this verifies that UCB's or SW-UCB's upper regret bound cannot increase slower than linear. Next, UCB's and SW-UCB's regret upper bound increases faster than either IDEA's or Kalman-UCB's regret upper bound found in Theorem \ref{theorem:regret_bound}, inequality \eqref{eq:Optimistic_Regret_Bound}. This is because the error of the statistic $\hat{\mu}_a$ is much larger than the error of the statistic $\left\langle a, \hat{z}_{t|t-1}\right\rangle$.

\section{Numerical Results}\label{sec:Random_Numerical_Comparisons}

For this section, we compare Kalman-UCB (Algorithm \ref{alg:Kalman-UCB}) and IDEA (Algorithm \ref{alg:IDEA}) with Kalman filter Observer Dependent Exploration (KODE) in \cite{gornet2025explorationfreemethodlinearstochastic} and a number of well-known SMAB algorithms. KODE is similar to Kalman-UCB and IDEA but selects actions that align most closely with the Kalman filter state prediction $\hat{z}_{t|t-1}$. For the set of well-known SMAB algorithms, we will compare our two proposed algorithms with UCB (Algorithm \ref{alg:UCB}) proposed by Auer, Cesa-Bianchi, and Fischer in \cite{auer2002finite} and SW-UCB (Algorithm \ref{alg:SW-UCB}) proposed by Garivier and Moulines in \cite{garivier2008upper}. Since our proposed environment samples rewards from a stationary distribution when the state matrix $\Gamma$ eigenvalues are within the unit circle, these are comparable algorithms. Next, we will compare the algorithms with Rexp3 proposed by Besbes and Zeevi in \cite{besbes2014stochastic}, which has proposed a general nonstationary bandit algorithm that addresses environments where the expected reward changes linearly. Finally, since the reward is the inner product of an action vector and an LGDS state variable, we added the linear bandit algorithm OFUL proposed by Abbasi-Yadkori, P\'{a}l, and Szepesv\'{a}ri in \cite{NIPS2011_e1d5be1c}. 

For the LGDS environment in \eqref{eq:LGDS}, we will generate the system parameters and noise statistics from a set of distributions where $k = d = 10$. Each parameter and statistic is independently sampled. For the noise statistic variance, note that $Q = RR^\top$ and $\sigma^2 = r^2$, where $R \sim p$ and $r \sim p$. For the state matrix $\Gamma \in \mathbb{R}^{d\times d}$, we first sampled a matrix $T \sim p$, $T \in \mathbb{R}^{d\times d}$, where each matrix entry of $T$ is independently sampled from the distribution $p$. We then normalize $T$ such that its eigenvalues are within the sphere of length $0.9$, i.e. $\Gamma = \left(0.9/\rho\left(T\right)\right)T$ where $\rho\left(T\right)$ is the spectral radius of matrix $T$. The distributions and their statistics are based on Table \ref{table:parameters}. 

For each distribution of Table \ref{table:parameters}, we generate $10^3$ different LGDS. Each algorithm interacts with the sampled LGDS $10$ different times for an interaction length of $n = 10^3$. Each LGDS state was initialized by computing the LGDS for $10^4$ iterations. In Table \ref{table:final_regret_median}, we have show the fractional difference of regret increased by each method with respect to the \textit{Kalman Oracle Action-selection} method (Algorithm \ref{alg:Kalman_Oracle}). In the table, IDEA (Algorithm \ref{alg:IDEA}), Kalman-UCB (Algorithm \ref{alg:Kalman-UCB}), and KODE \cite{gornet2025explorationfreemethodlinearstochastic} are significantly better than the other compared methods, where the medians plus their IQR's are still lower than the other method's median values for all the distributions besides the Cauchy distribution. This is because the statistic used for predicting the reward $X_t$ in Kalman-UCB and IDEA have significantly lower errors than the other methods. Finally, IDEA's median performance is the best across all the methods while also obtaining the lowest IQR values. 

\begin{table}[t]
    \centering
    \caption{Distributions}
    \label{table:parameters}
    \begin{tabular}{ll}
        \toprule
        \textbf{Distribution} & \textbf{Definition} \\
        \midrule
        Gaussian              & $\mathcal{N}(0,1)$ \\
        Uniform               & $[0,1]$ \\
        Exponential           & $\exp(1)$ \\
        Cauchy                & $X/Y$, $X, Y \sim \mathcal{N}(0,1)$ \\
        Bernoulli             & $P(X=1) = P(X=0) = 0.5$ \\
        \bottomrule
    \end{tabular}
\end{table}

% This table shows the median regret and interquartile range for each method
% across all environment types. Complements boxplot visualizations.
% With 10,000 data points per cell, median/IQR provides robust statistics.
\begin{table}[htbp]
\centering
\label{table:final_regret_median}
\caption{Normalized Regrets}
\resizebox{\linewidth}{!}{\begin{tabular}{lccccc}
\toprule
Method & Gaussian & Cauchy & Uniform & Bernoulli & Exponential \\
\midrule
IDEA & 1.37 (0.86) & 1.82 (8.25) & 0.84 (0.43) & 0.11 (0.08) & 0.08 (0.07) \\
KODE & 1.41 (0.88) & 1.84 (8.33) & 0.88 (0.45) & 0.11 (0.09) & 0.08 (0.07) \\
Kalman UCB & 1.52 (0.95) & 2.40 (12.44) & 0.90 (0.45) & 0.44 (0.22) & 0.57 (0.26) \\
OFUL & 3.94 (2.99) & 7.78 (25.95) & 1.79 (1.16) & 2.66 (1.40) & 3.30 (1.70) \\
Random Agent & 3.95 (2.99) & 7.86 (25.82) & 1.82 (1.14) & 2.90 (1.51) & 3.34 (1.73) \\
Rexp3 & 3.95 (2.98) & 7.85 (25.87) & 1.82 (1.14) & 2.87 (1.51) & 3.32 (1.73) \\
UCB & 3.84 (3.10) & 7.73 (25.47) & 1.72 (1.18) & 2.71 (1.46) & 3.16 (1.70) \\
\bottomrule
\end{tabular}} \\ [0.5em]
\sffamily 
\begin{flushleft}
Values are fractional difference between compared method and \textit{Kalman Oracle Action-selection} (Algorithm \ref{alg:Kalman_Oracle}). Higher values implies that the method's performance is worsening. Table uses statistic Median + (IQR) where IQR is the difference between the third quantile and the first quantile.
\end{flushleft}
\end{table}

\subsection{Numerical Comparisons of the Kalman-UCB versus IDEA}

In this section, we focus our analysis on the two methods: Kalman-UCB (Algorithm \ref{alg:Kalman-UCB}) and IDEA (Algorithm \ref{alg:IDEA}), to better understand the different exploration methodologies used by each method. In addition, it gives us more intuition about the metrics we derived in subsection \ref{sec:performance_comparison_theory}. The environments we use are discussed earlier in this section found in Table \ref{table:parameters}. 

Figure \ref{figure:regret} is a scatter plot where each dot compares the normalized regret values of Kalman-UCB and IDEA (each normalized regret value is a percentage of \textit{Kalman Oracle Action-selection}'s regret). The dashed red line indicates that the regret values for Kalman-UCB and IDEA are comparable. Dots above the red line imply that IDEA is performing better than Kalman-UCB and vice versa. Note that the axes are in logarithmic scale. 

In the figure, each plot is based on the distributions introduced in Table \ref{table:parameters}. Observe that for the Gaussian, Cauchy, and Uniform distributions, Kalman-UCB's and IDEA's normalized regrets are close to the dashed red line. This implies that the performance of each method is comparable. However, for the other distributions, IDEA performs consistently better than Kalman-UCB.  

\begin{figure}[t]
    \centering
    \includegraphics[width=0.7\linewidth]{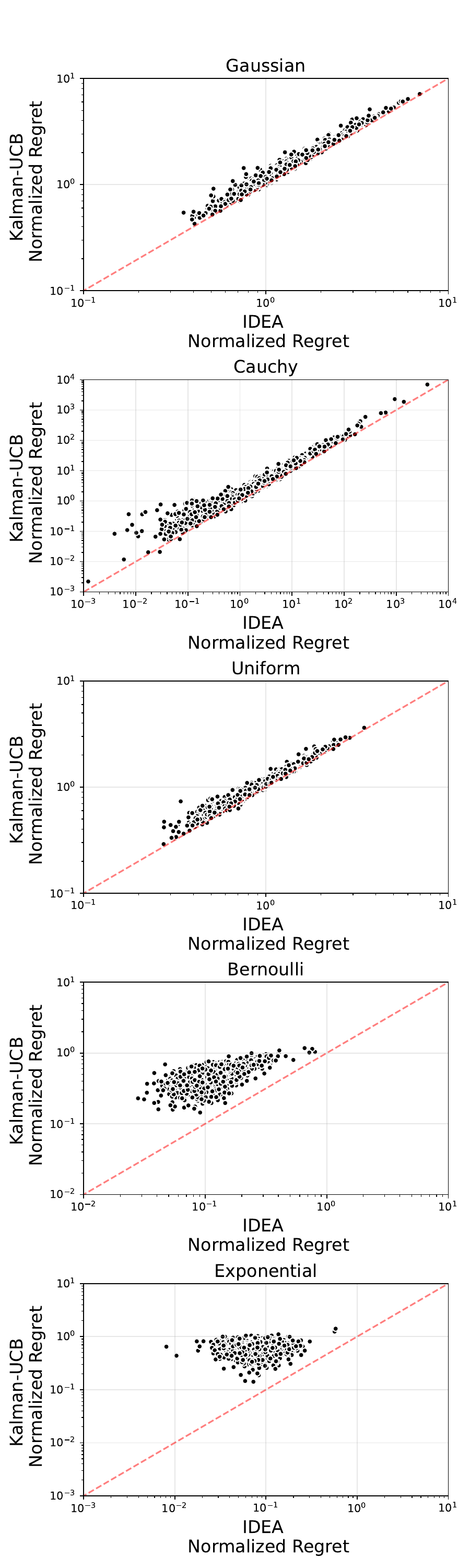}
    \caption{Scatter plot of the normalized regret values of Kalman-UCB versus IDEA. Note that the normalized regret value is the percentage of each algorithm's regret with respect to the \textit{Kalman Oracle Action-selection}'s regret. }
    \label{figure:regret}
\end{figure}

\subsection{Using the Metric to Quantify Performance}

In Section \ref{sec:Discussion}, Subsection \ref{sec:performance_comparison_theory}, a metric for comparing the performance of Kalman-UCB and IDEA was provided. This metric can be used to predict which method will perform better. Figure \ref{figure:comparison_theory} is a scatter plot where each red dot represents the lower bound of the interval while each blue dot represents the upper bound of the intervals. The dashed black line indicates that the lower/upper bound interval is comparable between the two methods.

Based on Figure \ref{figure:comparison_theory}, both the red and blue dots for the Bernoulli and Exponential distributions are above the dashed black line. If we observe Figure \ref{figure:regret}, the dots are consistently above the red line. However, for the Gaussian, Cauchy, and Uniform distributions in Figure \ref{figure:comparison_theory}, the upper bound blue dots are consistently close to the dashed black line. We can observe in Figure \ref{figure:regret} that the black dots are on the dash red line.  Therefore, the intervals help us predict which method will perform better, and we can observe that the upper interval gives a better indication of which algorithm will perform better. 

\begin{figure}[htbp]
    \centering
    \includegraphics[width=0.7\linewidth]{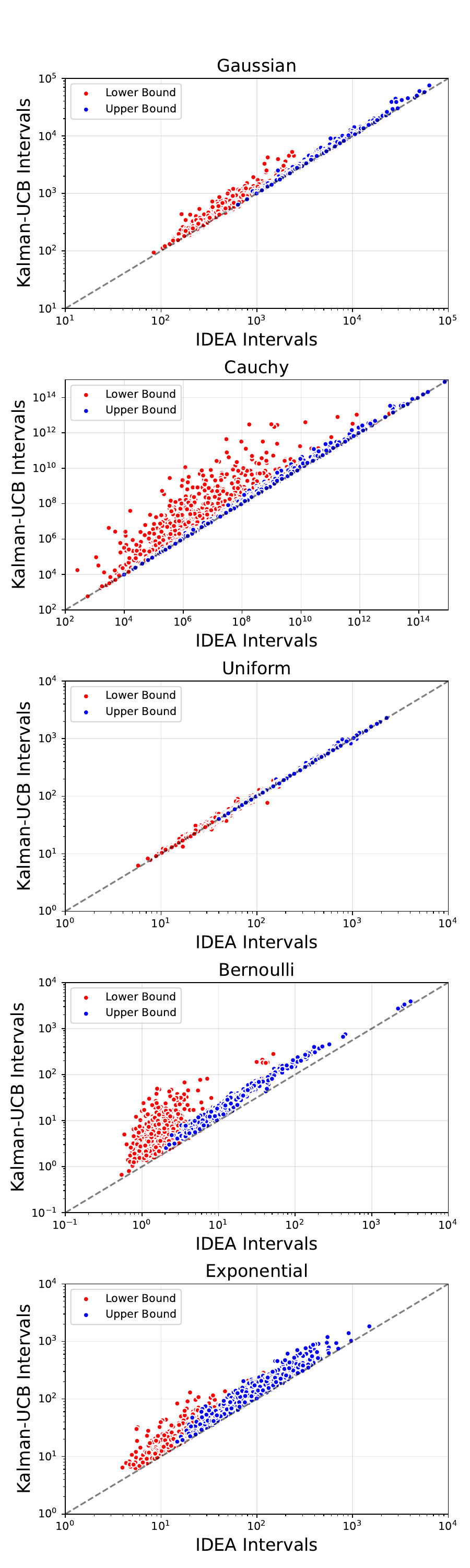}
    \caption{Scatter plot of Kalman-UCB's and IDEA's intervals \eqref{eq:interval_performance}. Blue dots are the upper bound and red dots are the lower bounds. }
    \label{figure:comparison_theory}
\end{figure}

\subsection{Robustness of KODE, IDEA, and Kalman-UCB}

For the final numerical analysis, we will be analyzing the robustness of KODE, IDEA, and Kalman-UCB. Recall that KODE, IDEA, and Kalman-UCB require prior knowledge of the system parameters $\Gamma$ and actions $a \in \mathcal{A}$ and the noise statistics $Q \succeq \mathbf{0}$ and $\sigma > 0$. In many cases, we would be required to identify these parameters and estimate the noise statistics, implying that there will be a degree of error of the identified parameters and estimates. Therefore, we will analyze the normalized regret of each method where the matrices and vectors used by KODE, IDEA, and Kalman-UCB are perturbed. Note that the \textit{Kalman Oracle Action-selection} will use unperturbed matrices and vectors

For each of the matrices and vectors, we first generate a matrix $\Xi$ where each component of the matrix is sampled from a normal distribution. Next, the matrix $\Xi$ is normalized such that $\Xi \gets \Xi/\left\Vert \Xi\right\Vert_F$, where $\left\Vert \cdot\right\Vert_F$ is the Frobenius norm. A matrix $T \gets I_d + \nu \Xi$ is defined, where $I_d$ is the identity matrix with dimension $d$ and $\nu \in \left\{0.1,1,10\right\}$ is a scaling factor. Finally, each matrix is set such that 
\begin{equation}
    \begin{array}{ccc}
        \tilde{\Gamma} \gets T^{-1} \Gamma T,  & \tilde{Q} \gets T^{-1} Q T, & \tilde{C_{\mathcal{A}}} \gets T^{-1} C_{\mathcal{A}} T
    \end{array} \nonumber, 
\end{equation}
where recall that $C_{\mathcal{A}}$ stacks the action vectors (see \eqref{eq:c_mathcal_a}). For each figure, we only perturb one matrix to understand which perturbations are the most impactful. 

Figure \ref{figure:robustness_gamma} is a box plot of KODE's, Kalman-UCB's, and IDEA's normalized regrets. The top row of subplots perturbs matrix $\Gamma$, the middle row of subplots perturbs actions $a \in \mathcal{A}$, and the bottom row perturbs matrix $Q \succeq \mathbf{0}$. The performance of the methods degrade most at noise magnitude $\nu = 10$ for the top and bottom rows, which are perturbations in the system parameters. In addition, the quantiles increase when the noise magnitudes increase to $\nu = 10$ for the same subplots. When comparing the changes in performance if matrix $\Gamma$ is perturbed, there is a $9\%$ decrease in median performance for KODE, a $18\%$ decrease in median performance for IDEA, and a $23\%$ decrease in median performance for Kalman-UCB. As for the actions $a \in \mathcal{A}$, there is a $47\%$ decrease in median performance for KODE, a $48\%$ decrease in median performance for IDEA, and a $40\%$ decrease in median performance for Kalman-UCB. Therefore, KODE is robust to changes of the matrix $\Gamma$ but is sensitive to changes in the actions $a \in \mathcal{A}$, while the opposite is true for Kalman-UCB. Finally, we can observe that IDEA has lower median regret across all the methods except for the case when the state matrix $\Gamma$ is perturbed with a noise magnitude of $\nu = 10$, which is the case where KODE performs best. 

\begin{figure}[h]
    \centering
    \includegraphics[width=\linewidth]{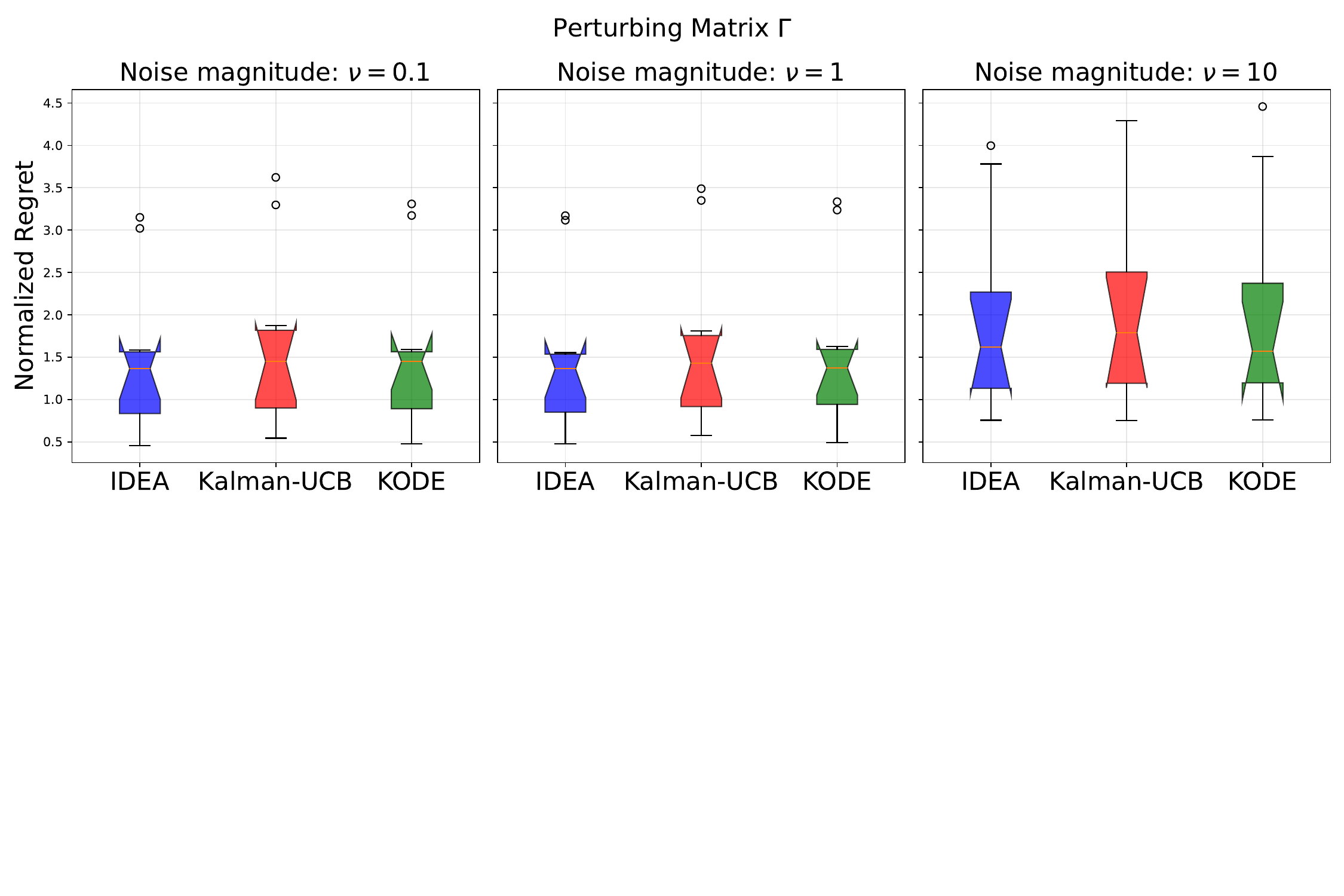}
    \includegraphics[width=\linewidth]{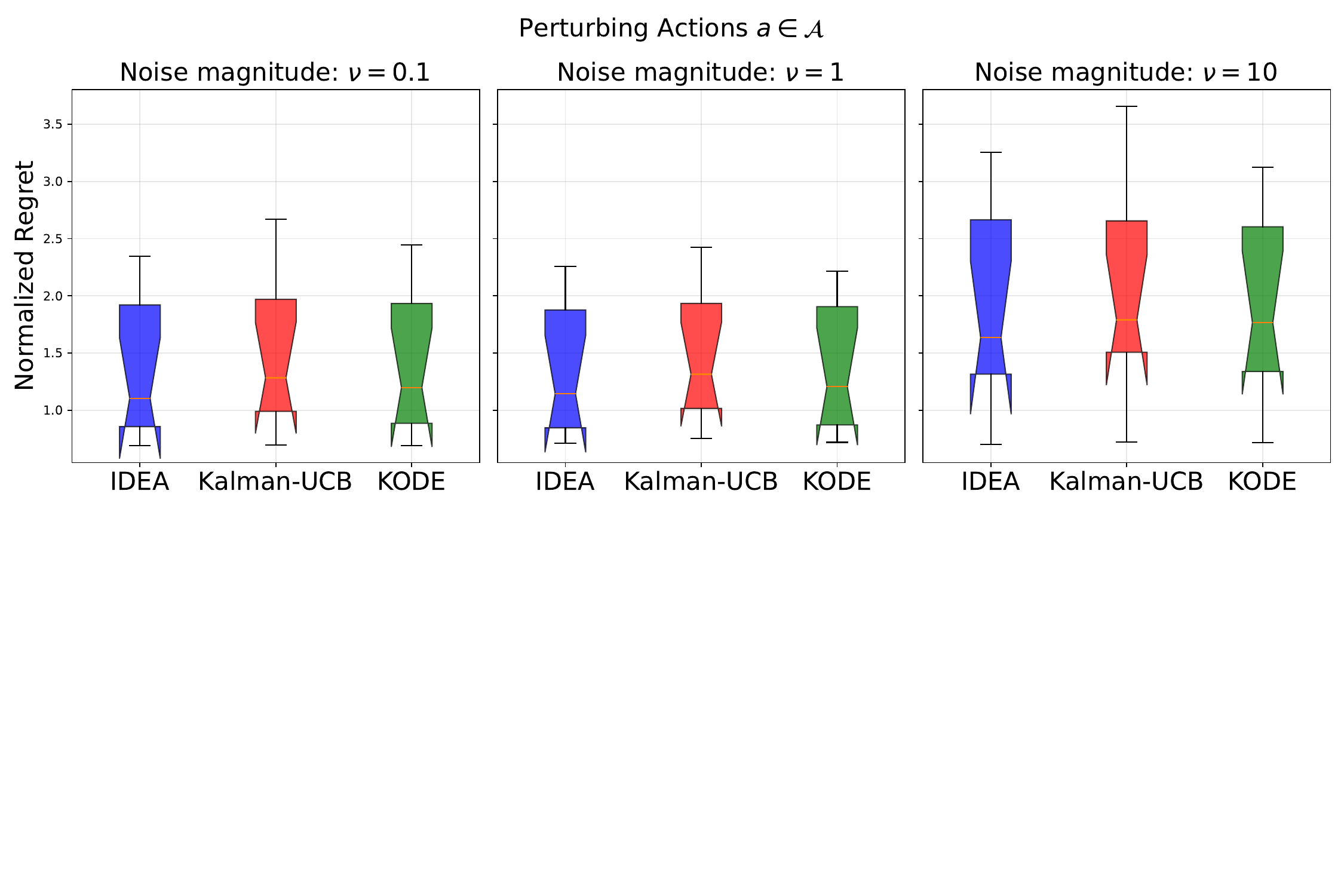}
    \includegraphics[width=\linewidth]{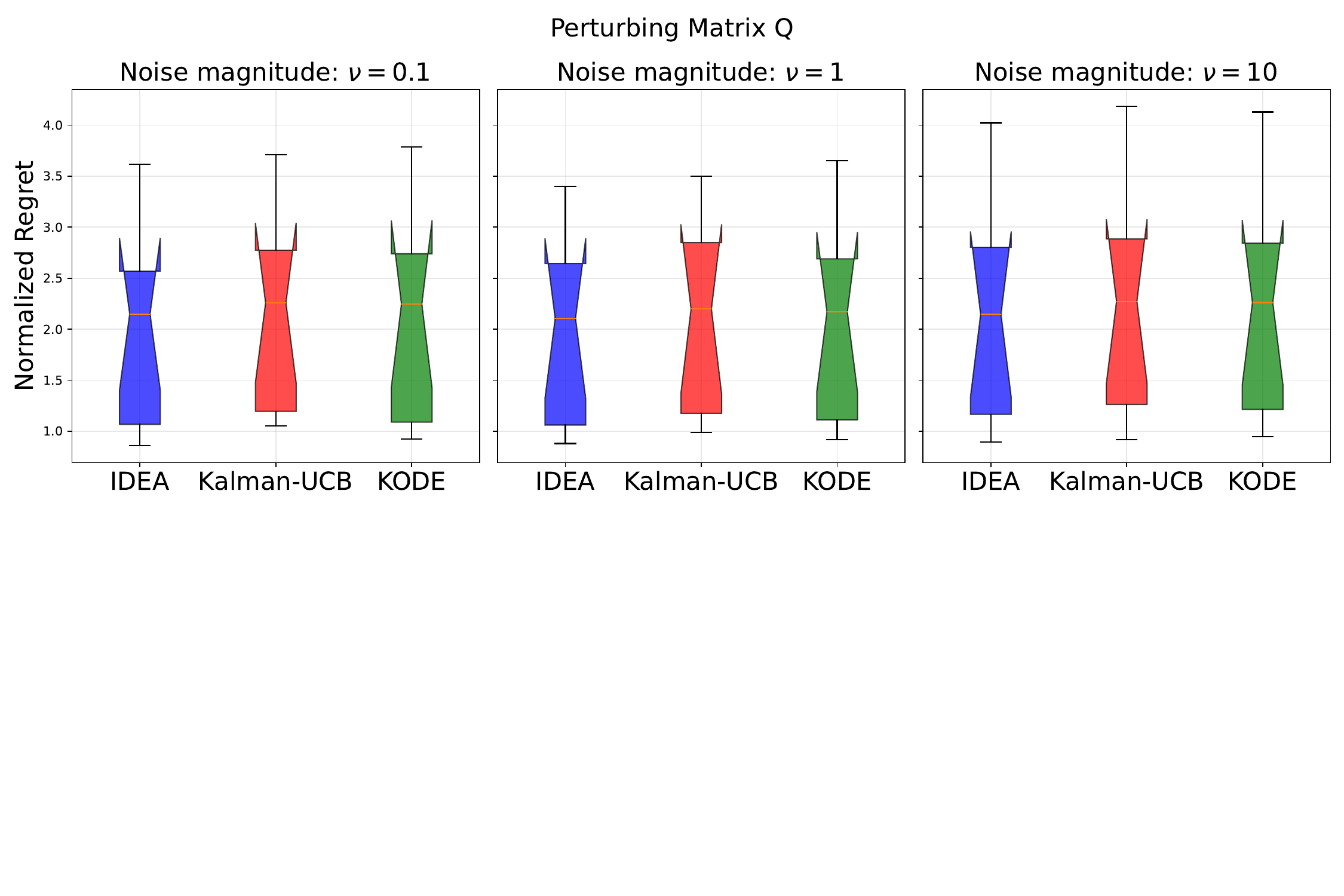}
    \caption{Box plot of each method's normalized regret. Each subplot is a different perturbation magnitude level.}
    \label{figure:robustness_gamma}
\end{figure}

\section{Conclusion}\label{sec:Conclusion}

In this paper, we studied the exploration-exploitation trade-off in a linear bandit environment where the reward is the output a Linear Gaussian Dynamical System (LGDS). The key contribution of this work are two methods: Kalman filter Upper Confidence Bound (Kalman-UCB) and Information filter Directed Exploration Action-selection (IDEA). Kalman-UCB selects actions that maximize the combination of the predicted reward and a term proportional to the error of the reward prediction. For IDEA, this method selects actions that maximize the combination of the predicted reward and a term proportional to how much the action minimizes the error of the Kalman filter's state prediction. Through theoretical analysis, we provided a metric to predict the relative performance between Kalman-UCB and IDEA and verified the results with numerical experiments across various random environments. Our findings suggest that IDEA, which accounts for information feedback in its perturbation term, may outperform Kalman-UCB in LGDS environments with an observable action.

\bibliographystyle{IEEEtran}
\bibliography{IEEEabrv,autosam}{}

\end{document}